\documentclass[reprint,superscriptaddress,amsmath,amssymb,aps,pra]{revtex4-2}
\usepackage{enumitem}

\usepackage{graphicx}
\usepackage{bm}
\usepackage[colorlinks=true,
  linkcolor=black,
  citecolor=blue,
  urlcolor =blue]{hyperref}
\usepackage{amsmath}
\usepackage{amssymb}
\usepackage{amsthm}
\usepackage{graphicx}
\usepackage[linesnumbered,ruled,vlined]{algorithm2e}
\usepackage{mathrsfs}
\usepackage{braket}
\usepackage[dvipsnames]{xcolor}
\usepackage{soul,xcolor}

\usepackage{booktabs}
\usepackage{threeparttable}
\usepackage{multirow}
 \usepackage[capitalize]{cleveref}
 \Crefname{appendix}{App.}{Apps.}
\Crefname{equation}{Eq.}{Eqs.}
\Crefname{figure}{Fig.}{Figs.}
\Crefname{section}{Sec.}{Secs.}

\newtheorem{thm}{\protect\theoremname}
\theoremstyle{plain}
\newtheorem{lem}{\protect\lemmaname}
\theoremstyle{plain}

\theoremstyle{plain}
\newtheorem*{lem*}{\protect\lemmaname}
\theoremstyle{plain}
\newtheorem*{thm*}{\protect\theoremname}
\theoremstyle{plain}
\newtheorem{prop}{\protect\propositionname}
\theoremstyle{plain}
\newtheorem{cor}{\protect\corollaryname}
\theoremstyle{plain}
\newtheorem*{cor*}{\protect\corollaryname}

\newtheorem*{defn*}{Definition}

\usepackage[USenglish]{babel}
  \providecommand{\corollaryname}{Corollary}
  \providecommand{\lemmaname}{Lemma}
  \providecommand{\propositionname}{Proposition}
  \providecommand{\remarkname}{Remark}
\providecommand{\theoremname}{Theorem}

\usepackage[normalem]{ulem}

\newcommand{\Or}{\mathcal{O}}
\newcommand{\RR}{\mathbb{R}}

\newcommand{\dd}{\mathrm{d}}

\renewcommand{\Re}{\operatorname{Re}}

\renewcommand{\ket}[1]{\ensuremath{\left|#1\right\rangle}}
\renewcommand{\bra}[1]{\ensuremath{\left\langle#1\right|}}

\newcommand{\CoefLambda}{\boldsymbol{\lambda}}
\newcommand{\CoefMu}{\boldsymbol{\mu}}
\newcommand{\BasisBeta}{\boldsymbol{\beta}}

\begin{document}
\title{Learning Hamiltonians in the Heisenberg limit with static single-qubit fields}








\author{Shrigyan Brahmachari}
\thanks{These authors contributed equally to this work.}
\affiliation{Department of Electrical and Computer Engineering, Duke University, Durham, NC 27708, USA}
\author{Shuchen Zhu}
\thanks{These authors contributed equally to this work.}
\affiliation{Department of Mathematics, Duke University, Durham, NC 27708, USA}
\author{Iman Marvian}
\affiliation{Department of Electrical and Computer Engineering, Duke University, Durham, NC 27708, USA}
\affiliation{Department of Physics, Duke University, Durham, NC 27708, USA}
\affiliation{Duke Quantum Center, Duke University, Durham, NC 27701, USA}
\author{Yu Tong}
\affiliation{Department of Mathematics, Duke University, Durham, NC 27708, USA}
\affiliation{Department of Electrical and Computer Engineering, Duke University, Durham, NC 27708, USA}
\affiliation{Duke Quantum Center, Duke University, Durham, NC 27701, USA}

\begin{abstract}
 Learning the Hamiltonian governing a quantum system is a central task in quantum metrology, sensing, and device characterization. 
Existing Heisenberg-limited Hamiltonian learning protocols either require multi-qubit operations that are prone to noise, or single-qubit operations whose frequency or strength increases with the desired precision. These two requirements limit the applicability of Hamiltonian learning on near-term quantum platforms. We present a protocol that learns a quantum Hamiltonian with the optimal Heisenberg-limited scaling using only single-qubit control in the form of static fields with strengths that are independent of the target precision. Our protocol is robust against the state preparation and measurement (SPAM) error. By overcoming these limitations, our protocol provides new tools for device characterization and quantum sensing. We demonstrate that our method achieves the Heisenberg-limited scaling through rigorous mathematical proof and numerical experiments. We also prove an information-theoretic lower bound showing that a non-vanishing static field strength is necessary for achieving the Heisenberg limit unless one employs an extensive number of discrete control operations. 
\end{abstract}

\maketitle


 A broad range of problems in quantum metrology and sensing can be formulated as learning the properties of an unknown Hamiltonian $H$ that governs the dynamics of a closed quantum system, through measurements of its time evolution
\cite{degen2017quantum,giovannetti2011advances,giovannetti2006quantum,caves1981quantum,leibfried2004toward,de2005quantum,GranadeFerrieWiebeCory2012robust,  wiebe2014b, HuangTongFangSu2023learning,StilckFranca2024,higgins2007entanglement}.
For example, a qubit magnetometer \cite{degen2017quantum} can be used to measure an unknown external magnetic field $\vec{B}$ by probing the evolution generated by the Hamiltonian
$H=-\mu\vec{B}\cdot\vec{\sigma}$, which describes a spin-$1/2$ system with magnetic moment $\mu$ in the magnetic field $\vec{B}$.
By monitoring the system over time, one can extract information about the unknown Hamiltonian $H$, and hence about the magnetic field $\vec{B}$.
To learn such an unknown Hamiltonian with precision $\epsilon$, one must probe the system for a time that scales at least as $\epsilon^{-1}$, a bound reminiscent of energy-time uncertainty relations and commonly referred to as the Heisenberg limit.

However, even for the practically important cases of single- and two-qubit Hamiltonians with multiple unknown parameters,   the simple protocol consisting of state preparation, free evolution under $H$, and measurement typically achieves only the standard quantum limit, requiring a total evolution time $T=\mathcal{O}(\epsilon^{-2})$ to reach error $\epsilon$ \cite{PangBrun2014quantum}.
Achieving Heisenberg-limited scaling for such general Hamiltonians therefore requires some form of \textit{quantum control}.
For instance, one may apply control operations to effectively reshape or simulate the Hamiltonian \cite{HuangTongFangSu2023learning}, or introduce multi-qubit interactions that modify the Hamiltonian \cite{dutkiewicz2024advantage,bakshi2024structure,HaahKothariODonnellTang2023query}. However, these techniques come with their own challenges. For example, the Hamiltonian-reshaping technique in \cite{HuangTongFangSu2023learning} requires applying single-qubit operations at an increasingly high rate: to achieve error $\epsilon$, the interval between successive gate applications must scale no larger than $\sqrt{\epsilon}$. Similarly, in the approaches of \cite{dutkiewicz2024advantage,bakshi2024structure,HaahKothariODonnellTang2023query}, multi-qubit operations are required to partially cancel the effect of the unknown Hamiltonian or unitary. Such operations are typically prone to errors, which in turn limit the attainable precision.

\begin{figure}
\centering
\includegraphics[width=7cm]{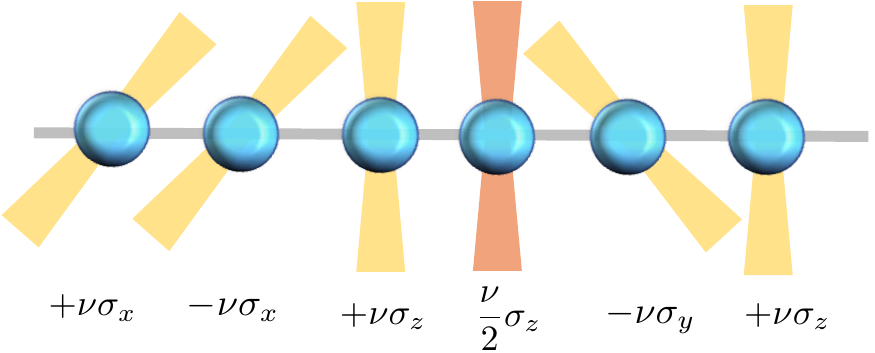}
\caption{\textbf{Efficient Hamiltonian learning with static local fields.} For a system of qubits with unknown Hamiltonian $H(\CoefLambda)$, by adding sufficiently strong static single-qubit fields along the $\pm\hat{x}$, $\pm\hat{y}$, and $\pm\hat{z}$ directions, and choosing a different field strength on a single qubit (highlighted in a different color in the figure, namely a field strength $\nu/2$ instead of $\nu$), we obtain a Hamiltonian whose ground and first excited states are nondegenerate and sufficiently close to unentangled states, denoted by $|\Phi_0\rangle$ and $|\Phi_1\rangle$, respectively.
By preparing the system in $\frac{1}{\sqrt{2}}(\ket{\Phi_0}+\ket{\Phi_1})$, which is also an unentangled state, and evolving it under the total Hamiltonian, we can efficiently extract information about the Hamiltonian terms in $H(\CoefLambda)$. Repeating this procedure for different qubits and field directions allows us to reconstruct the original Hamiltonian $H(\CoefLambda)$ with arbitrary precision $\epsilon$, using a total evolution time that saturates the Heisenberg limit $1/\epsilon$. 
}
\label{fig:control_ham}
\end{figure}
In this Letter, we introduce a simple and practical Hamiltonian learning protocol that achieves the optimal Heisenberg-limited scaling, enabling the estimation of a quantum Hamiltonian with error $\epsilon$ using a total evolution time $T=\mathcal{O}(\epsilon^{-1})$.   
Crucially, unlike previous approaches (e.g., Ref.~\cite{HuangTongFangSu2023learning}), our protocol requires only single-qubit control in the form of static fields whose strengths are independent of the target precision. Specifically, static control fields of fixed strength $\nu>0$ are applied along the $x$, $y$, and $z$ directions, modifying the system Hamiltonian from $H$ to $H-\nu H_{\mathrm{ctrl}}$  (see Fig.~\ref{fig:control_ham}). 

Furthermore, the protocol does not require any entangling gates, even when the Hamiltonian $H$ itself is entangling. The qubits are initially prepared in a tensor product of Pauli eigenstates, evolve under the modified Hamiltonian $H-\nu H_{\mathrm{ctrl}}$, and are finally measured via single-qubit Pauli measurements. Finally, the protocol is robust against state-preparation and measurement (SPAM) errors: as long as these errors remain below a fixed threshold independent of $\epsilon$, the target precision can still be achieved. Theorem~\ref{thm:ham_learn} characterizes the performance of the protocol. 

\begin{thm}
\label{thm:ham_learn}
For an unknown Hamiltonian given in \eqref{eq:ham_to_be_learned_main_text} on $n=\Or(1)$ qubits, we can learn all its coefficients with $\ell^2$-error at most $\epsilon$ with probability at least $1-\delta$ through a non-adaptive protocol that satisfies the following properties:
    \begin{enumerate}
    \item It uses $\Or(\epsilon^{-1}\log(1/\delta))$ total evolution time (Heisenberg-limited scaling).
    \item It consists of $\Or(\mathrm{polylog}(\epsilon^{-1}\delta^{-1}))$  non-adaptive experiments, each using product state inputs and single-qubit measurements.
    \item It uses only single-qubit control (i.e., static field such as described by $-\nu H_{\mathrm{ctrl}}$ in \eqref{eq:H_ctrl_example}) in $\hat{x},\hat{y},\hat{z}$ directions. The strength $\nu$ is independent of $\epsilon,\delta$.
    \item The protocol is robust against state preparation and measurement (SPAM) error.
\end{enumerate}
\end{thm}
A proof sketch is given after the protocol description, with full details in the Supplemental Material (SM).

\vspace{1em}

To demonstrate the performance of our protocol, we focus on the practically important cases of $n=1$ and $n=2$ qubits, which are relevant for applications ranging from magnetometry to quantum computing. We evaluate the protocol through extensive numerical simulations. Fig.~\ref{fig:error_vs_time} shows results for single- and two-qubit Hamiltonians, where the Heisenberg-limited scaling is evident for sufficiently large $\nu$ ($\nu \geq 1.9$ for the single-qubit case and $\nu \geq 4.5$ for the two-qubit case) in the presence of SPAM noise. The unknown single-qubit Hamiltonian is chosen to be $H=0.1 \sigma^x + 0.5\sigma^y + 0.3\sigma^z$. More details of the simulations, including the two-qubit Hamiltonian and the initial guesses used in optimization, are provided in the Section~VI in SM, where we also numerically demonstrate the SPAM robustness of our protocol. 
\begin{figure}[t]
\centering

\begin{minipage}[t]{0.50\linewidth}
  \centering
  \includegraphics[width=\linewidth,trim=6pt 6pt 6pt 6pt,clip]{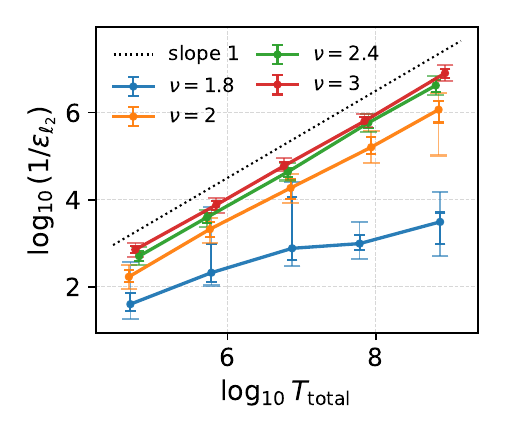}
\end{minipage}\hspace{-0.015\linewidth}%
\begin{minipage}[t]{0.50\linewidth}
  \centering
  \includegraphics[width=\linewidth,trim=6pt 6pt 6pt 6pt,clip]{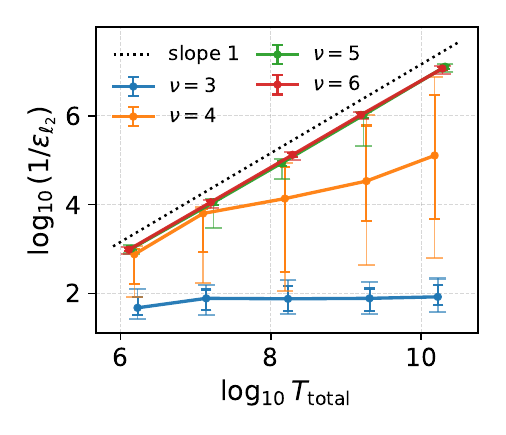}
\end{minipage}

\caption{\textbf{Learning single-qubit and two-qubit Hamiltonians. } Whisker plots showing the $\ell^2$-error $\varepsilon_{\ell^2}$ with different total evolution times $T$ and values of static field strengths $\nu$ for learning a single-qubit Hamiltonian (left) and a two-qubit Hamiltonian (right) on a log scale. The inner whiskers represents the 35th–65th percentiles, with the dot indicating the median. The outer whiskers extend from the 25th to the 75th percentiles. Dotted lines with slope 1 are included as references for the Heisenberg-limited scaling. 
Independent bit-flip channels with error rates 0.05 (for single-qubit case) and 0.03 (for two-qubit case) are applied to each qubit before measurement to model SPAM noise.
}
\label{fig:error_vs_time}
\end{figure}



\vspace{1em}
Our second main result, which complements the above theorem, is a generalization of the so-called Heisenberg limit. It establishes a fundamental lower bound on the minimum total evolution time $T$ required to achieve a target error $\epsilon$. More precisely, we derive a quantitative trade-off between the achievable precision $\epsilon$ and the relevant physical and operational resources, namely the field strength $\nu$, the total evolution time $T$, and the number of \emph{discrete control operations} $L$, defined as the number of times the system’s evolution is interrupted in order to measure, prepare, or reset the qubits, or to modify the control Hamiltonian. Crucially, our bound fully accounts for the possibility of adaptive learning protocols. We note that related extensions of the Heisenberg limit have previously been established in settings without static fields but with discrete control gates \cite{dutkiewicz2024advantage,ma2024learningkbodyhamiltonianscompressed}.

\begin{thm}[Informal version of Theorem~S6 in SM]
\label{thm:lower_bound}
Suppose there exists a (possibly adaptive) protocol for learning an unknown single-qubit Hamiltonian
$H=\hat{n}\cdot\vec{\sigma}$, with $\|\hat{n}\|=1$, to precision $\epsilon$ using
$L=o(1/\epsilon)$ discrete control operations. Suppose the protocol is allowed to apply control Hamiltonians of the form $-\nu H_{\mathrm{ctrl}}$,
where $H_{\mathrm{ctrl}}$ is an arbitrary single-qubit Hermitian operator with operator norm
at most $1$.
Then, for generic Hamiltonians $H$, any such protocol must use a total evolution time at least 
$
T=\Omega\left(\frac{1}{4\nu\epsilon+3\epsilon^2}\right).
$
\end{thm} 

In the above we write $f=\Omega(g)$ to mean {$f(x)\geq cg(x)$ for some constant $c>0$ and sufficiently large $x$.}
A proof is included in Section VII in SM.
In particular, this theorem shows that without an extensive number of discrete control operations, i.e., for $L\ll 1/\epsilon$, one can only achieve the Heisenberg-limited scaling with static field strength $\nu=\Omega(1)$, and $1/\nu$ also lower bounds the prefactor appearing in the $1/\epsilon$ scaling.  
We need to exclude the situation where $\Omega(1/\epsilon)$ discrete control operations are used because that alone would be sufficient for attaining the Heisenberg limit \cite{dutkiewicz2024advantage,bakshi2024structure}. This lower bound reveals a fundamental tradeoff between the total evolution time, the field strength, and discrete control operations needed to attain a given precision.

To prove this theorem, we consider the task of discriminating between two fixed single-qubit Hamiltonians corresponding to two known normalized vectors, e.g., $\hat{z}$ and $(\hat{z}+\epsilon \hat{x})/\sqrt{1+\epsilon^2}$, a problem previously studied in Ref.~\cite{PangBrun2014quantum}.  If there exists a learning protocol that learns them with precision $\epsilon/2$ then we will be able to distinguish these two Hamiltonians.    
Our key observation, proven in Corollary~S1 in SM, is that in the presence of the control Hamiltonian $\nu H_{\mathrm{ctrl}}$, 
the distance between the unitary evolution operators generated by these two Hamiltonians grows like $\Or((\nu \epsilon + \epsilon^2)T + L\epsilon)$. We adopt the framework used in \cite{HuangTongFangSu2023learning,huang2022foundations} to represent any adaptive experimental protocol as a tree, and analyze the distance between output distributions generated by these two Hamiltonians, thereby generating a lower bound through Le Cam's two-point method \cite{yu1997assouad}.

\vspace{1em}

\noindent{\emph{The Setup}--} In the rest of this Letter, we describe and analyze our learning protocol in further detail. We consider a system of $n$ qubits evolving under an unknown Hamiltonian
\begin{equation}
    \label{eq:ham_to_be_learned_main_text}
    H(\CoefLambda) = \sum_{\mathbf{a}\in \{0,1,2,3\}^n\setminus 0^n} \lambda_{\mathbf{a}} \sigma^{\mathbf{a}},
\end{equation}
where  $
\sigma^{\mathbf{a}} = \sigma^{a_1}\otimes \sigma^{a_2}\otimes\cdots\otimes \sigma^{a_n}
$.  Here $\sigma^0$ denotes the single-qubit identity, while $\sigma^{1,2,3}$ denote the Pauli $X$, $Y$, and $Z$ matrices, respectively (We will also sometimes use $\sigma^{x,y,z}_k$ to denote the Pauli-X, Y, or Z operator acting on the $k$th qubit). Without loss of generality, we assume that the Hamiltonian is traceless, and to fix the normalization, we assume the unknown coefficients $\lambda_{\mathbf{a}}$ satisfy $-1\leq \lambda_{\mathbf{a}}\leq 1$ 
  (Any Hamiltonian can be rescaled so that this condition holds).


Our goal is to obtain estimates $\hat{\CoefLambda}=(\hat{\lambda}_{\mathbf{a}})$ such that with probability, at least, $1-\delta$, 
its $\ell^2$-distance with actual 
$\CoefLambda$ is bounded by $\epsilon$, i.e.,   
\begin{equation}
    \Pr[\|\hat{\CoefLambda}-\CoefLambda\|\leq \epsilon]\geq 1- \delta.
\end{equation}

For the control Hamiltonian, we will apply sufficiently strong static single-qubit fields along the $\pm\hat{x}$, $\pm\hat{y}$, and $\pm\hat{z}$ directions. More specifically, for $n\ge 2$, we choose the field strength on the $k$th qubit to be $\nu/2$ and that on the remaining qubits to be $\nu$. As an example, Figure~\ref{fig:control_ham} shows a control Hamiltonian
\begin{equation}
\label{eq:H_ctrl_example}
    \nu H_{\mathrm{ctrl}} = \nu\sigma^x_1-\nu \sigma^x_2 + \nu \sigma^z_3 + \frac{\nu}{2}\sigma^z_4 - \nu\sigma^y_5 + \nu\sigma^z_6
\end{equation}
Any such Hamiltonian can be specified by a set of parameters $(k,\mathbf{s},\BasisBeta)$, where $\mathbf{s}$ is a bit-string specifying the signs ($0,1$ for $+,-$ respectively), and $\BasisBeta\in \{1,2,3\}^n$ specifies the directions. In the above example $\mathbf{s}=010010$, $\BasisBeta=(1,1,3,3,2,3)$, and $k=4$. The system then evolves under the total Hamiltonian
\begin{equation}
\label{eq:total_ham}
    H_{\mathrm{tot}}(\CoefLambda,\nu,k,\mathbf{s},\BasisBeta) = H(\CoefLambda)-\nu H_{\mathrm{ctrl}}(k,\mathbf{s},\BasisBeta).
\end{equation}


This choice of the control Hamiltonian gaurantees that the ground state $\ket{\Psi_0}$ and the first excited state $\ket{\Psi_1}$ of $H_{\mathrm{tot}}(\CoefLambda,\nu,k,\mathbf{s},\BasisBeta)$ are both nondegenerate and protected by energy gaps scaling with $\nu$ (That is why we have chosen a different field strength of $\nu/2$ on one of the qubits.). Moreover, they are close to the ground state $\ket{\Phi_0}$ and first excited state $\ket{\Phi_1}$ of $-\nu H_{\mathrm{ctrl}}(k,\mathbf{s},\BasisBeta)$ respectively. Also, our construction ensures that $\ket{\Phi_0}$, $\ket{\Phi_1}$, and their superposition  $\ket{\Phi_+} = \frac{1}{\sqrt{2}}(\ket{\Phi_0}+\ket{\Phi_1})$
are unentangled product  state of  single-qubit Pauli eigenstates, and are therefore easy to prepare.

\vspace{1em}
\noindent\emph{Phase estimation experiment.} To learn the coefficients $\CoefLambda$, we perform a sequence of \emph{phase-estimation experiments} to estimate the energy gap $E_{\Delta}(\CoefLambda,\nu,k,\mathbf{s},\BasisBeta)$ between the ground and first excited states of $H_{\mathrm{tot}}(\CoefLambda,\nu,k,\mathbf{s},\BasisBeta)$. As explained below, this can be achieved despite the fact that $H$ is unknown and the relevant eigenstates of $H_{\mathrm{tot}}$ and their superpositions cannot be prepared.  

The key idea is to prepare the qubits in the state $\ket{\Phi_+} $, i.e., the superposition of the ground and first excited states of $-H_{\mathrm{ctrl}}$.  We then let the system evolve under the modified Hamiltonian $H_{\mathrm{tot}}=H-\nu H_{\mathrm{ctrl}}$ for time $t$, and then measure one of two observables $O_c$ and $O_s$, each of which is a single-qubit Pauli operator on qubit $k$. Specifically, they are chosen such that $O_c |\Phi_0\rangle=|\Phi_1\rangle$, and 
$O_s |\Phi_0\rangle=-i |\Phi_1\rangle$, which means they act as $\sigma^x$ and $-\sigma^y$ in the subspace spanned by $\ket{\Phi_0}$ and $\ket{\Phi_1}$ (Note that $\ket{\Phi_0}$ and $\ket{\Phi_1}$ are identical on all qubits, except qubit $k$).
As an example, in the setting described in Figure~\ref{fig:control_ham}, with $\beta_k=3$, we can simply choose $O_c= \sigma^x_k$ and $O_s=-\sigma^y_k$, which are Pauli operators acting on the $k$th qubit.

Measuring the Pauli operator $O$ on qubit $k$, we observe $ (-1)^b: b=0,1$ eigenvalue, where bit $b=0$ with probability
\begin{equation}
\label{eq:output_probabilities_main_text}
   \Pr[b=0]=  \frac{1}{2}\left(1+\braket{\Phi_+|e^{iH_{\mathrm{tot}}t}O e^{-iH_{\mathrm{tot}}t}|\Phi_+}\right),
\end{equation}
where $O$ stands for $O_c$ or $O_s$, depending on which we choose to measure.  Each run yields a sample from this distribution. From this statistic, one can obtain estimates of the expectation values $\braket{O_c(t)}$ and $\braket{O_s(t)}$ as functions of time $t$, for the initial state $|\Phi_+\rangle$.

Next, applying perturbation theory, in SM we show that for  sufficiently large $\nu$ these expectation values are  close to the expectation values obtained for the initial state $|\Psi_+\rangle=(|\Psi_0\rangle+|\Psi_1\rangle)/\sqrt{2}$, the equal superposition of the actual ground and first excited states of $H_{\mathrm{tot}}$. Hence, the dominant frequency in signals $\braket{O_c(t)}$ and $\braket{O_s(t)}$ is determined by the energy gap $E_\Delta$ between these two eigenstates. Indeed, we show in SM that the error is bounded by 
\begin{equation}
    \label{eq:Oc_Os_err_bound}
    |\braket{O_c(t)}+i\braket{O_s(t)} - e^{iE_\Delta t}|=\Or(\|H\|/\nu).
\end{equation}
Crucially, as long as this error remains bounded by a finite constant, one can extract the dominant frequency $E_\Delta$ with Heisenberg-limited precision using a robust frequency-estimation protocol \cite{LiTongNiGefenYing2023heisenberg,ma2024learningkbodyhamiltonianscompressed,HuMaGongEtAl2025,MobusBluhmCaroEtAl2023dissipation}, closely related to the robust phase-estimation algorithm of Ref.~\cite{KimmelLowYoder2015robust}.   
More precisely, to estimate $E_\Delta$ to precision $\epsilon$, one only needs to estimate every $\braket{O_c(t)}$ and $\braket{O_s(t)}$ to a precision independent of $\epsilon$, namely within $1/\sqrt{8}$ additive error. Indeed, we show that using $\Or(\mathrm{polylog}(\epsilon^{-1}\delta^{-1}))$ experiments that run in time at most $\Or(1/\epsilon)$, one obtains an $\epsilon$-accurate estimate of $E_\Delta$. The total evolution time needed for this series of experiments is $\Or(\epsilon^{-1}\log(1/\delta))$.

Moreover, the errors in estimates for $\braket{O_c(t)}$ and $\braket{O_s(t)}$ are not required to be unbiased. Therefore, even if there is a systematic bias in our estimates created by the SPAM error, our protocol can still achieve arbitrarily high precision in the estimation for $E_\Delta$. In this sense our protocol is robust against SPAM error.

\vspace{1em}

\noindent\emph{Learning protocol.} 
The next step is to learn the unknown Hamiltonian coefficients $\CoefLambda$ from the collected information about the energy gaps $E_{\Delta}(\CoefLambda,\nu,k,\mathbf{s},\BasisBeta)$. Note that, a priori, it is not clear whether $\CoefLambda$ can be (i) uniquely determined from $E_{\Delta}(\CoefLambda,\nu,k,\mathbf{s},\BasisBeta)$, and (ii) determined robustly, i.e., if the energy gaps are estimated with precision $\epsilon'$, the recovery map $g$ estimates the coefficients to precision at most $C\epsilon'$, for some constant $C$ independent of $\CoefLambda$. For instance, this condition fails if $g(x)=\sqrt{x}$, since achieving precision $\epsilon$ in estimating $x\ll 1$ results in a quadratically worse precision, $\sqrt{\epsilon}$. In the present context, such behavior would prevent achieving the Heisenberg limit.

We prove that these requirements are all satisfied by the least-square minimizer, defined below. Let $\mathbf{E}(\CoefLambda,\nu)=(E_{\Delta}(\CoefLambda,\nu,k,\mathbf{s},\BasisBeta))_{k,\mathbf{s},\BasisBeta}$ be the vector of all actual energy gaps corresponding to all experiments, and  $\hat{\mathbf{E}}$ be the corresponding estimates obtained through the above method. Define the loss function
\begin{equation}
    \label{eq:loss_function_main_text}
    L(\hat{\mathbf{E}},\CoefMu,\nu) = \frac{1}{2}\|\hat{\mathbf{E}}-\mathbf{E}(\CoefMu,\nu)\|^2\ ,
\end{equation}
where  $\CoefMu$ is a free variable. 
Then, we choose our estimate of the Hamiltonian coefficient as $\hat{\CoefLambda} = \mathrm{argmin}_{\CoefMu} L(\hat{\mathbf{E}},\CoefMu,\nu)$. 



Note that since the problem is nonlinear, in general, it may contain multiple local minima, making it difficult to uniquely recover the correct $\CoefLambda$. However, thanks to the choice of our control Hamiltonian $H_{\mathrm{ctrl}}$, in SM we are able to prove the following crucial properties in SM:

\begin{enumerate}[itemsep=2pt, leftmargin=*]
 \item  When $\CoefMu$ and $\hat{\mathbf{E}}$ are within a constant distance of the true values $\CoefLambda$ and $\mathbf{E}^\star=\mathbf{E}(\CoefLambda,\nu)$, respectively, 
the cost function $L(\hat{\mathbf{E}},\CoefMu,\nu)$ is locally strongly convex. In particular, 
 all eigenvalues of the Hessian matrix $\partial_{\CoefMu}^2 L(\hat{\mathbf{E}},\CoefMu,\nu)$ are lower bounded by  $2^{n-1}$   (See Proposition~S1 in SM).
 
 \item Using a modified version of the inverse function theorem proved in the SM, this strong convexity guarantees that, for $\hat{\mathbf{E}}$ sufficiently close to $\mathbf{E}^\star$, there exists a unique $\mu$ that minimizes the cost function, which we take as the estimate $\hat{\CoefLambda}$. That is, there exists a map $g$ such that
 \begin{equation}
    \hat{\CoefLambda}=g(\hat{\mathbf{E}})=\mathrm{argmin}_{\CoefMu} L(\hat{\mathbf{E}},\CoefMu,\nu).
 \end{equation}
 This map can be computed using standard numerical optimization methods, such as Newton’s method.
 \item Our modified inverse function theorem provides a constant upper bound on $\|\partial_{\hat{\mathbf{E}}}g(\hat{\mathbf{E}})\|$ in the neighborhood of $\mathbf{E}^\star$. This, in particular, implies that a perturbation of magnitude $\epsilon'$ on {$\hat{\mathbf{E}}$ }translates to a perturbation at most $C\epsilon'$ on $\hat{\CoefLambda}$, {for a fixed constant $C$.} 
\end{enumerate}

In conclusion, these results imply that recovering the coefficients $\CoefLambda$ to precision $\epsilon$ only requires estimating $\mathbf{E}^\star$ to precision $\epsilon/C$, thus preserving the Heisenberg-limited scaling. Since the strong convexity plays a key role in the effectiveness of our method, we provide some intuition. The Hessian matrix can be computed as







\begin{equation}
    \label{eq:hessian_LS}
    \frac{\partial^2 L}{\partial\CoefMu^2} = \left(\frac{\partial \mathbf{E}}{\partial\CoefMu}\right)^\top \left(\frac{\partial \mathbf{E}}{\partial\CoefMu}\right) + (\mathbf{E}(\CoefMu)-\hat{\mathbf{E}})^\top \frac{\partial^2\mathbf{E}}{\partial\CoefMu^2}\ .
\end{equation}
When $\CoefMu$ is close to the actual value $\CoefLambda$ and $\hat{\mathbf{E}}$ close to $\mathbf{E}^\star$, both $\mathbf{E}(\CoefMu)$ and $\hat{\mathbf{E}}$ will be close to $\mathbf{E}^\star$, making the second term on the right-hand side small. Furthermore, as we show in SM, the first term is
\begin{equation}
    \left(\frac{\partial \mathbf{E}}{\partial\CoefMu}\right)^\top \left(\frac{\partial \mathbf{E}}{\partial\CoefMu}\right) = J_0 + \Or_n(1/\nu)\ ,
\end{equation}
where $J_0$ is the Hessian matrix in the $\nu\to \infty$ limit for $\hat{\mathbf{E}}=\mathbf{E}(\CoefMu)$. Therefore, for large enough $\nu$ the smallest eigenvalue of the Hessian matrix $\partial_{\CoefMu}^2 L$ can be well-approximated by that of $J_0$. In Section~III of SM, we explicitly compute the matrix $J_0$ using the fact that in $\nu\to\infty$ limit, the eigenvalues of the Hamiltonian and coefficients $\CoefLambda$ are related through a Walsh-Hadamard transform, also known as the Fourier transform over the Boolean group $(\mathbb{Z}_2)^{\times n}$.. 

\vspace{1em}

\noindent\emph{Error and time analysis.} 
In the SM, we analyze the error in the protocol and prove Theorem~\ref{thm:ham_learn}. To estimate the coefficients $\CoefLambda$ with precision $\epsilon$, we first obtain an estimate $\hat{\mathbf{E}}$ of the eigenvalues via phase-estimation experiments with precision $\Or(\epsilon)$. This requires a total evolution time $\mathcal{O}(1/\epsilon)$, achieving Heisenberg-limited scaling. The estimate $\hat{\CoefLambda}=g(\hat{\mathbf{E}})$ is then obtained by solving the optimization problem defined by \eqref{eq:loss_function_main_text}. Theorem~S2 in the SM guarantees that this solution is close to the true value of $\CoefLambda$. The map $g(\hat{\mathbf{E}})$ is computed using the iterative solver in Algorithm~1 of the SM, which requires an initial guess within a constant distance of the true parameters. Such an initial guess can be obtained using a suboptimal Hamiltonian learning protocol \cite{StilckFranca2024} that does not employ quantum control, and instead relies on state preparation, time evolution under the Hamiltonian $H$, and measurement. Since this preliminary step only requires constant precision, its non-Heisenberg-limited scaling does not affect the overall performance. The protocol is fully non-adaptive, and the control operations depend only on the number of qubits. 

\vspace{1em}
\noindent \emph{Discussions.}  In this work, we introduce a Heisenberg-limited Hamiltonian learning protocol that overcomes practical limitations of existing approaches. Unlike prior methods that rely either on multi-qubit operations or on single-qubit control whose strength or frequency scales with the desired estimation precision, our protocol achieves optimal Heisenberg-limited scaling using only single-qubit control implemented through static fields of fixed, precision-independent strength. This makes the method well suited for near-term quantum devices with limited control capabilities. In terms of the analysis of the algorithm, we note that it goes beyond a rotating wave approximation (e.g., as described in \cite{Burgarth2022oneboundtorulethem,BurgarthFacchiEtAl2024taming}) because we cannot allow any approximation error to accumulate in time.

Our protocol applies to an arbitrary Hamiltonian on $n$ qubits. For generic Hamiltonians, the presence of $4^n-1$ independent coefficients implies an exponential runtime in $n$. In contrast, for Hamiltonians with only local interactions, the cost scales polynomially, or even polylogarithmically, with the system size $n$   \cite{li2020hamiltonian,che2021learning,HaahKothariTang2022optimal,ZubidaYitzhakiEtAl2021optimal,BaireyAradEtAl2019learning, bairey2020learning,GranadeFerrieWiebeCory2012robust,gu2022practical,wilde2022learnH,KrastanovZhouEtAl2019stochastic,caro2022learning,MobusBluhmCaroEtAl2023dissipation,HolzapfelEtAl2015scalable, HuangTongFangSu2023learning,dutkiewicz2024advantage,MiraniHayden2024learning,NiLiYing2024quantum,LiTongNiGefenYing2023heisenberg,BoixoSomma2008parameter,bakshi2024structure,MobusBluhmCaroEtAl2023dissipation,ma2024learningkbodyhamiltonianscompressed,yu2023robust,hangleiter2024robustlylearninghamiltoniandynamics,StilckFranca2024,HuMaGongEtAl2025,SinhaTong2025improved,Rosati2025quantum,Zhao2025learning}. It is then natural to ask whether one can generalize the ideas in this work to the setting of local Hamiltonians with a $\mathrm{poly}(n)$ or $\mathrm{polylog}(n)$ cost. 

Besides Hamiltonian learning, the idea of extracting information by applying single-qubit static fields may be useful for related tasks such as learning a function of Hamiltonian parameters \cite{
abbasgholinejad2025optimally, eldredge2018optimal, proctor2017networked, proctor2018multiparameter,suzuki2020quantum,ehrenberg2023minimum,bringewatt2024optimal} without learning the whole Hamiltonian. In this setting, ideas from Hamiltonian learning such as Hamiltonian reshaping has already proved useful \cite{abbasgholinejad2025optimally}. One may also consider the setting in which only part of the quantum system is accessible, as studied in  \cite{ChenCotlerHuang2025quantum}, and static fields may provide a useful tool to decouple and to probe the system. One may also consider leveraging techniques from quantum error correction to mitigate decoherence in the quantum system \cite{Zhou2017AchievingTH}.



\acknowledgements 

The authors thank Di Fang for comments in the process of this work. This work is supported by a collaboration between the US DOE and other Agencies. This material is based upon work supported by the U.S. Department of Energy, Office of Science, National Quantum Information Science Research Centers, Quantum Systems Accelerator. SZ acknowledges support from the U.S. Department of Energy, Office of Science, Accelerated Research in Quantum Computing Centers, Quantum Utility through Advanced Computational Quantum Algorithms, grant No. DE-SC0025572. Additional support is
acknowledged from  
NSF QLCI grant OMA-2120757.

\bibliography{ref_new}


\clearpage

\renewcommand{\thefigure}{S\arabic{figure}}
\setcounter{thm}{0}
\renewcommand{\thethm}{S\arabic{thm}}
\setcounter{lem}{0}
\renewcommand{\thelem}{S\arabic{lem}}
\setcounter{prop}{0}
\renewcommand{\theprop}{S\arabic{prop}}
\setcounter{cor}{0}
\renewcommand{\thecor}{S\arabic{cor}}

\onecolumngrid

\section*{Notations}

We will use the following notation for Pauli matrices:
\begin{equation}
    \sigma^0 = \begin{pmatrix}
        1 & 0 \\
        0 & 1
    \end{pmatrix},\quad 
    \sigma^1 = \begin{pmatrix}
        0 & 1 \\
        1 & 0
    \end{pmatrix},\quad 
     \sigma^2 = \begin{pmatrix}
        0 & -i \\
        i & 0
    \end{pmatrix},\quad 
     \sigma^3 = \begin{pmatrix}
        1 & 0 \\
        0 & -1
    \end{pmatrix}.
\end{equation}
For an $n$-qubit quantum system, we will write $\sigma^a_{j}$ for $\sigma^a$ acting only on the $j$th qubit.
Let $\mathbf{a}=(a_1,a_2,\cdots,a_n)\in\{0,1,2,3\}^n$, we denote
\[
\sigma^{\mathbf{a}} = \sigma^{a_1}\otimes \sigma^{a_2}\otimes\cdots\otimes \sigma^{a_n}=\prod_{j=1}^n \sigma^{a_j}_{j}.
\]
For Pauli eigenstates, we use $\ket{s,a}$ to denote the $(-1)^s$-eigenstate of $\sigma^a$. In other words,
\[
\ket{0,1}=\ket{+},\quad \ket{1,1} = \ket{-},\quad \ket{0,2}=\ket{+i},\quad \ket{1,2} = \ket{-i}, \quad \ket{0,3}=\ket{0},\quad \ket{1,3} = \ket{1}.
\]
For two strings $\mathbf{a},\mathbf{b}\in\{0,1,2,3\}^n$, we say they are \textit{compatible}, denoted by $\mathbf{a}\sim\mathbf{b}$, if for every $j\in\{1,2,\cdots,n\}$, one of the following is true
\begin{equation}
\label{eq:compatible_strings}
    a_j=b_j,\ a_j=0,\text{ or } b_j=0. 
\end{equation}
If $\mathbf{a}\sim\mathbf{b}$, then $\sigma^{\mathbf{a}}$ and $\sigma^{\mathbf{b}}$ commute in every component. 
For any $a\in\{0,1,2,3\}$, we define
\begin{equation}
    \label{eq:s_function_scalar}
    s(a) = \begin{cases}
        0 \text{ if } a=0,\\
        1 \text{ if } a\neq 0.
    \end{cases}
\end{equation}
For $\mathbf{a}\in\{0,1,2,3\}^n$, we also define
\begin{equation}
    \label{eq:s_function}
    s(\mathbf{a}) = (s(a_1),s(a_2),\cdots,s(a_n)).
\end{equation}
We write 
\begin{equation}
    |\mathbf{a}|=\sum_j s(a_j)
\end{equation}
to denote the number of non-zero elements in $\mathbf{a}$.

We adopt the asymptotic notation: $f=\Or(g)$ means $|f|\leq Cg$ for some absolute constant $C$; $f=\Omega(g)$ means $g=\Or(f)$, and $f=\Or_n(g)$ means the constant we omit may depend on $n$. 

\section{Problem setup}
\label{sec:problem_setup}

We assume access to a quantum system evolving under the unknown Hamiltonian
\begin{equation}
    \label{eq:ham_to_be_learned}
    H(\CoefLambda) = \sum_{\mathbf{a}\in \{0,1,2,3\}^n} \lambda_{\mathbf{a}} \sigma^{\mathbf{a}},
\end{equation}
where $-1\leq \lambda_{\mathbf{a}}\leq 1$, and $\CoefLambda=(\lambda_{\mathbf{a}})_{\mathbf{a}\neq \mathbf{0}}$. We assume that the Hamiltonian is traceless so as to avoid the ambiguity coming from the global phase. Our goal is to obtain estimates $\hat{\lambda}_{\mathbf{a}}$ for each $\mathbf{a}\neq \mathbf{0}$ such that
\begin{equation}
    \Pr[\|\hat{\CoefLambda}-\CoefLambda\|\leq \epsilon]\geq 1-\delta,
\end{equation}
Here $\|\cdot\|$ denotes the $\ell^2$-norm, $\epsilon$ is the target precision and $\delta$ is the failure probability we can tolerate.

We will probe the system by applying external fields in the form of single-qubit Hamiltonians. More precisely, add to the original Hamiltonian a control Hamiltonian $-\nu H_{\mathrm{ctrl}}(k,\mathbf{s},\BasisBeta)$ 
\begin{equation}
\label{eq:H_ctrl}
    H_{\mathrm{ctrl}}(k,\mathbf{s},\BasisBeta) = \sum_{j\neq k}  (-1)^{s_j} \sigma^{\beta_j}_j + \frac{1}{2}(-1)^{s_k} \sigma^{\beta_k}_k,
\end{equation}
where $\mathbf{s}=(s_1,\cdots,s_n)\in\{0,1\}^n$ and $\BasisBeta=(\beta_1,\cdots,\beta_n)\in\{1,2,3\}^n$. The system then evolves under the Hamiltonian 
\begin{equation}
\label{eq:total_ham}
    H_{\mathrm{tot}}(\CoefLambda,\nu,k,\mathbf{s},\BasisBeta) = H(\CoefLambda)-\nu H_{\mathrm{ctrl}}(k,\mathbf{s},\BasisBeta).
\end{equation}

\section{The phase estimation experiment}
\label{sec:the_phase_est_experiment}

For simplicity we will temporarily omit the dependence on $\nu,\CoefLambda,\BasisBeta,k,\mathbf{s}$ since these are fixed within a single experiment. 
We define the following states:  
\begin{equation}
    \label{eq:eigenstates_of_H_ctrl}
    \ket{\Phi_0} = \bigotimes_{j} \ket{s_j,\beta_j},\quad \ket{\Phi_1} = \left(\bigotimes_{j<k}\ket{s_j,\beta_j}\right)\otimes \ket{1-s_k,\beta_k}\otimes \left(\bigotimes_{j>k}\ket{s_j,\beta_j}\right).
\end{equation}
One can readily verify that they are the ground state and the first excited state of $-H_{\mathrm{ctrl}}$ respectively. In our phase estimation experiment, we will initialize the quantum system in the state
\begin{equation}
    \label{eq:initial_state}
    \ket{\Phi_+} = \frac{1}{\sqrt{2}}(\ket{\Phi_0}+\ket{\Phi_1}).
\end{equation}
Note that this state is still a tensor product of single-qubit Pauli eigenstates and is therefore easy to prepare. We then let the quantum system evolve under the unknown Hamiltonian $H$ together with the external control $-\nu H_{\mathrm{ctrl}}$, which can be described by the time evolution operator $e^{-iH_{\mathrm{tot}}t}$. At the end of the experiment when measure one of the two single-qubit Pauli observables $O_c$ and $O_s$. They are chosen so that 
\begin{equation}
    \label{eq:Oc_Os_requirements}
    O_c\ket{\Phi_0} = \ket{\Phi_1},\quad O_c\ket{\Phi_1}=\ket{\Phi_0},\quad O_s\ket{\Phi_0} = -ii\ket{\Phi_1},\quad O_s\ket{\Phi_1}=i\ket{\Phi_0}.
\end{equation}
Such $O_c$ and $O_s$ are always available in the form of single-qubit Pauli operators.
As an example, when $\beta_k=3$, we can simply choose $O_c= \sigma^1_k$ and $O_s=-\sigma^2_k$, i.e., Pauli-$X$ and $Y$ operators acting on the $k$th qubit. Therefore from each run of the experiment we sample a binary output $b$ such that,
\begin{equation}
\label{eq:output_probabilities}
    \Pr[b=0] = \frac{1}{2}\left(1+\braket{\Phi_+|e^{iH_{\mathrm{tot}}t}O_{c/s} e^{-iH_{\mathrm{tot}}t}|\Phi_+}\right),
\end{equation}
where the observable $O_{c/s}$ in the middle is $O_c$ or $O_s$ depending on which we choose to measure.

\subsection{Analysis of the phase estimation experiment}
\label{sec:analysis_phase_estimation}

In this subsection we will discuss what results we expect from the phase estimation experiment described above. Our initial state is $\ket{\Phi_+}$, which is an equal superposition of $\ket{\Phi_0}$ and $\ket{\Phi_1}$. Note that these two states are the ground state and the first excited state of $-H_{\mathrm{ctrl}}$ respectively, whereas the system evolves under $H_{\mathrm{tot}}=H-\nu H_{\mathrm{ctrl}}$.
We denote the ground state and the first excited state of this Hamiltonian by $\ket{\Psi_0}$ and $\ket{\Psi_1}$, with the corresponding eigenvalues $E_0$ and $E_1$.
Because $\ket{\Phi_0}$ and $\ket{\Phi_1}$ are eigenstates of $-H_{\mathrm{ctrl}}$ with energy gap $1$, we have the following lemma, which follows from applying Lemma~\ref{lem:eigenstate_perturbation_bound} to $(1/\nu)H_{\mathrm{tot}}=-H_{\mathrm{ctrl}}+H/\nu$:
\begin{lem}
    \label{lem:initial_state_eigenstate_difference}
    If $\nu \geq 3\|H\|$, with an appropriately chosen global phase for $\ket{\Psi_0}$ and $\ket{\Psi_1}$, we have
    \[
    \left\|\ket{\Psi_0}-\ket{\Phi_0}\right\|\leq 3\|H\|/\nu,\quad \left\|\ket{\Psi_1}-\ket{\Phi_1}\right\|\leq 3\|H\|/\nu.
    \]
\end{lem}
This means that $\ket{\Phi_+}$ is approximately a superposition of $\ket{\Psi_0}$ and $\ket{\Psi_1}$. We will use this observation to approximately compute the output probabilities given in \eqref{eq:output_probabilities}. It suffices to compute
\begin{equation}
    \Big|\braket{\Phi_+|e^{iH_{\mathrm{tot}}t}O_{c/s} e^{-iH_{\mathrm{tot}}t}|\Phi_+} - \frac{1}{2}(\bra{\Psi_0}e^{i E_0 t}+\bra{\Psi_1}e^{i E_1 t})O_{c/s}(e^{-i E_0 t}\ket{\Psi_0}+e^{-i E_1 t}\ket{\Psi_1})\Big|\leq 6\sqrt{2}\|H\|/\nu
\end{equation}
where we have used Lemma~\ref{lem:initial_state_eigenstate_difference} and the triangle inequality.
Using \eqref{eq:Oc_Os_requirements} we have 
\begin{equation}
\begin{aligned}
    \Big|\braket{\Phi_+|e^{iH_{\mathrm{tot}}t}O_{c} e^{-iH_{\mathrm{tot}}t}|\Phi_+} &- \cos((E_1-E_0)t)\Big|\leq 6\sqrt{2}\|H\|/\nu , \\
    \Big|\braket{\Phi_+|e^{iH_{\mathrm{tot}}t}O_{s} e^{-iH_{\mathrm{tot}}t}|\Phi_+} &- \sin((E_1-E_0)t)\Big| \leq 6\sqrt{2}\|H\|/\nu.
\end{aligned}
\end{equation}
Therefore, we define the energy gap
\begin{equation}
\label{eq:E_Delta}
    E_{\Delta} = E_1 - E_0,
\end{equation}
then, when we are measuring $O_c$, the probability of obtaining $0$ satisfies
\begin{equation}
\label{eq:cos_err_bound}
    \Big|\Pr[b=0] - \frac{1+\cos(E_\Delta t)}{2}\Big|\leq 6\sqrt{2}\|H\|/\nu,
\end{equation}
and when we are measuring $O_s$ the probability likewise satisfies
\begin{equation}
\label{eq:sin_err_bound}
    \Big|\Pr[b=0] - \frac{1+\sin(E_\Delta t)}{2}\Big|\leq 6\sqrt{2}\|H\|/\nu,
\end{equation}

\subsection{Robust frequency estimation}
\label{sec:robust_frequency_est}

The robust frequency estimation algorithm \cite{LiTongNiGefenYing2023heisenberg,ma2024learningkbodyhamiltonianscompressed,HuMaGongEtAl2025,MobusBluhmGefenEtAl2025heisenberg}, which is a variant of the robust phase estimation algorithm \cite{KimmelLowYoder2015robust}, allows us to estimate the energy gap $E_\Delta$ to precision $\epsilon$ with $1/\epsilon$ total evolution time, as stated in the following theorem:


\begin{thm}[Theorem~4.1 in \cite{MobusBluhmGefenEtAl2025heisenberg}]
\label{thm:robust_frequency_est}
    Let $\theta\in[-\Phi,\Phi]$.
    Let $X(t)$ and $Y(t)$ be independent random variables satisfying
    \begin{equation} \label{eq:rv_correctness_condition_prob}
        \begin{aligned}
            &|X(t)-\cos(\theta t)|< 1/\sqrt{8}, \text{ with probability at least }2/3, \\
            &|Y(t)-\sin(\theta t)|< 1/\sqrt{8}, \text{ with probability at least }2/3.
        \end{aligned}
    \end{equation}
    Then with independent samples $X(t_1),X(t_2),\cdots,X(t_{\Gamma})$ and $Y(t_1),Y(t_2),\cdots,Y(t_{\Gamma})$, with
    \begin{equation}
        {\Gamma}=\Or(\log^2(\Phi/\epsilon)), \quad T=t_1+t_2+\cdots+t_{\Gamma}=\Theta(1/\epsilon), \quad \max_j t_j=\Or(1/\epsilon) ,
    \end{equation}
    and $t_j\geq 0$, we can construct  an estimator $\hat{\theta}$ such that
    \begin{equation}
        \sqrt{\mathbb{E}[|\hat{\theta}-\theta|^2]}\leq \epsilon.
    \end{equation}
\end{thm}
The algorithm works by solving a sequence of
decision problems, in each of which we start with an interval $[a,b]$ containing $\theta$, and shrink it to a smaller interval containing $\theta$ that is either $[a,\frac{1}{3}a + \frac{2}{3}b]$ or $[\frac{2}{3}a + \frac{1}{3}b,b]$, thus refining our knowledge of the unknown parameter $\theta$. These decision problems can be correctly solved with evolution time proportional to $1/(b-a)$ with high probability using a few rounds of majority voting. This results in
an algorithm that is robust against noise and achieves the $1/\epsilon$ total evolution time (i.e., $T$) scaling.

In our setting, let $\hat{B}$ be the fraction of experiments yielding $b=0$ with $O_c$ or $O_s$. Then $2\hat{B}-1$ serves as an estimator of either $\cos(E_\Delta t)$ or $\sin(E_\Delta t)$ with bias at most $12\sqrt{2}\|H\|/\nu$ by \eqref{eq:cos_err_bound} and \eqref{eq:sin_err_bound}. The variance is at most $1/N_b$ when $N_b$ independent samples are used to compute the rate $\hat{B}$. Consequently one can estimate $\cos(E_\Delta t)$ or $\sin(E_\Delta t)$ up to error $12\sqrt{2}\|H\|/\nu + \sqrt{3/N_b}$ with probability at least $2/3$ as guaranteed by Chebyshev's inequality. By Theorem~\ref{thm:robust_frequency_est}, it suffices to have
\[
12\sqrt{2}\|H\|/\nu + \sqrt{3/N_b}\leq 1/\sqrt{8},
\]
for which we can choose $\nu = 96\|H\|$ and $N_b = 96$. Importantly, these choices are independent of the precision $\epsilon$ that we want to achieve.

Besides the expected error bound, we also want an error bound that holds with large probability. This can be obtained by taking the median of $N_m$ independent samples of $\hat{\theta}$ (i.e., the estimate for $\theta$, or in our setting $E_\Delta$, obtained from one run of the robust frequency estimation algorithm), which through Hoeffding's inequality ensures the median is within $3\epsilon$ distance from the actual value $E_\Delta$ with probability at least $1-e^{-cN_m}$ for an absolute constant $c$. Therefore a $1-\delta$ success probability can be guaranteed with $\Or(\log(1/\delta))$ overhead.

\section{Recovering the coefficients}

We regard $E_\Delta$ as a function of $\CoefLambda,\nu,k,\mathbf{s},\BasisBeta$, and write it as $E_\Delta(\CoefLambda,\nu,k,\mathbf{s},\BasisBeta)$. 
Because $\CoefLambda$ has been used to represent the actual coefficients of the unknown Hamiltonian, we introduce a free variable $\CoefMu$ to replace $\CoefLambda$ and consider the function $E_\Delta(\CoefMu,\nu,k,\mathbf{s},\BasisBeta)$
All energy gaps computed from $\CoefMu,\nu$ for all $k,\mathbf{s},\BasisBeta$ (which means all $n\times 3^n\times 2^n$ combinations) are collected into an array
\begin{equation}
\label{eq:eigenvalue_function}
    \mathbf{E}(\CoefMu,\nu) = (E_\Delta(\CoefMu,\nu,k,\mathbf{s},\BasisBeta))_{k,\mathbf{s},\BasisBeta}.
\end{equation}
We regard this array as a vector of dimension $6^n n$:

\subsection{The Jacobian matrix}

We will first study the Jacobian matrix of $\mathbf{E}(\CoefMu,\nu)$ as a function of $\CoefMu$. We will show that the Jacobian matrix 
\begin{equation}
    \label{eq:jacobian_matrix}
    J(\CoefMu,\nu) = \frac{\partial \mathbf{E}(\CoefMu,\nu)}{\partial\CoefMu},
\end{equation}
is well-conditioned. This requires us to first estimate the smallest singular value of $J(\CoefMu,\nu)$. $J(\CoefMu,\nu)$ is of size $6^n\times (4^n-1)$, where the $-1$ comes from the fact that the coefficient associated with $I$ in $H$ is set to be $0$ (so that $H$ is traceless).

Using the eigenvalue perturbation theory \cite{Bha97:Matrix-Analysis}, we have
\begin{equation}
    \frac{\partial E_\Delta(\CoefMu,\nu,k,\mathbf{s},\BasisBeta)}{\partial\lambda_{\mathbf{a}}} = \frac{1}{2}\big(\braket{\Psi_1(\CoefMu,\nu,k,\mathbf{s},\BasisBeta)|\sigma^{\mathbf{a}}|\Psi_1(\CoefMu,\nu,k,\mathbf{s},\BasisBeta)} - \braket{\Psi_0(\CoefMu,\nu,k,\mathbf{s},\BasisBeta)|\sigma^{\mathbf{a}}|\Psi_0(\CoefMu,\nu,k,\mathbf{s},\BasisBeta)}\big),
\end{equation}
where $\ket{\Psi_0(\CoefMu,\nu,k,\mathbf{s},\BasisBeta)}$ and $\ket{\Psi_1(\CoefMu,\nu,k,\mathbf{s},\BasisBeta)}$ are the ground state and the first excited state of $H_{\mathrm{tot}}(\CoefMu,\nu,k,\mathbf{s},\BasisBeta)$ respectively. Using Lemma~\ref{lem:initial_state_eigenstate_difference} again, we have
\begin{equation}
    \frac{\partial E_\Delta(\CoefMu,\nu,k,\mathbf{s},\BasisBeta)}{\partial\lambda_{\mathbf{a}}} = \frac{1}{2}\big(\braket{\Phi_1(k,\mathbf{s},\BasisBeta)|\sigma^{\mathbf{a}}|\Phi_1(k,\mathbf{s},\BasisBeta)} - \braket{\Phi_0(\mathbf{s},\BasisBeta)|\sigma^{\mathbf{a}}|\Phi_0(\mathbf{s},\BasisBeta)}\big) + \Or(\|H\|/\nu).
\end{equation}
Because both $\ket{\Phi_0}$ and $\ket{\Phi_1}$ are tensor products of single-qubit Pauli eigenstates, we can readily compute their Pauli expectation values: If $\mathbf{a}\sim\BasisBeta$, then
\begin{equation}
    \label{eq:Pauli_expectation_values_Phi}
    \braket{\Phi_0(\mathbf{s},\BasisBeta)|\sigma^{\mathbf{a}}|\Phi_0(\mathbf{s},\BasisBeta)} = (-1)^{\mathbf{s}\cdot s(\mathbf{a})},\quad \braket{\Phi_1(k,\mathbf{s},\BasisBeta)|\sigma^{\mathbf{a}}|\Phi_1(k,\mathbf{s},\BasisBeta)} = (-1)^{s(a_k)}(-1)^{\mathbf{s}\cdot s(\mathbf{a})}.
\end{equation}
If $\mathbf{a}\not\sim\BasisBeta$, then both expectation values are $0$. Here the symbol $\sim$ is defined in \eqref{eq:compatible_strings} and the function $s(\cdot)$ is defined in \eqref{eq:s_function_scalar} and \eqref{eq:s_function}. From the above, we can see that
\begin{equation}
    \frac{\partial E_\Delta(\CoefMu,\nu,k,\mathbf{s},\BasisBeta)}{\partial\lambda_{\mathbf{a}}} = \begin{cases}
        -(-1)^{\mathbf{s}\cdot s(\mathbf{a})} + \Or(\|H\|/\nu),&\text{ if } \mathbf{a}\sim\BasisBeta\text{ and } s(a_k)\neq 0, \\
        0 ,&\text{ otherwise.}
    \end{cases}
\end{equation}
Therefore we can compute the matrix $J(\CoefMu,\nu)^\top J(\CoefMu,\nu)$ in the following way
\begin{equation}
\begin{aligned}
    (J(\CoefMu,\nu)^\top J(\CoefMu,\nu))_{\mathbf{a},\mathbf{a}'} &= \sum_{k,\mathbf{s},\BasisBeta}  \frac{\partial E(\CoefMu,\nu,k,\mathbf{s},\BasisBeta)}{\partial\lambda_{\mathbf{a}}}  \frac{\partial E(\CoefMu,\nu,k,\mathbf{s},\BasisBeta)}{\partial\lambda_{\mathbf{a}'}} \\
    &= \sum_{\substack{k: s(a_k)=1\\ s(a_k')=1}} \sum_{\substack{\BasisBeta:\BasisBeta\sim\mathbf{a} \\ \BasisBeta\sim\mathbf{a}'}}\sum_{\mathbf{s}} (-1)^{\mathbf{s}\cdot (s(\mathbf{a})+s(\mathbf{a}'))} + \Or_n(1/\nu).
\end{aligned}
\end{equation}
Because 
\begin{equation}
    \sum_{\mathbf{s}} (-1)^{\mathbf{s}\cdot (s(\mathbf{a})+s(\mathbf{a}'))} = 2^n\delta_{s(\mathbf{a}),s(\mathbf{a}')}, 
\end{equation}
we have
\begin{equation}
    (J(\CoefMu,\nu)^\top J(\CoefMu,\nu))_{\mathbf{a},\mathbf{a}'} = 2^n\sum_{\substack{k: s(a_k)=1\\ s(a_k')=1}} \sum_{\substack{\BasisBeta:\BasisBeta\sim\mathbf{a} \\ \BasisBeta\sim\mathbf{a}'}}\delta_{s(\mathbf{a}),s(\mathbf{a}')} + \Or_n(1/\nu).
\end{equation}
Note that $\BasisBeta\sim \mathbf{a}$ and $\BasisBeta\sim\mathbf{a}'$ imply $\mathbf{a}\sim\mathbf{a}'$. Given the additional requirement that $s(\mathbf{a})=s(\mathbf{a}')$ for the term to be non-zero, we can then see that it is non-zero only when $\mathbf{a}=\mathbf{a}'$. Therefore
\begin{equation}
\label{eq:JTJ_explicit}
    \begin{aligned}
        (J(\CoefMu,\nu)^\top J(\CoefMu,\nu))_{\mathbf{a},\mathbf{a}'} &= 2^n\sum_{\substack{k: s(a_k)=1\\ s(a_k')=1}} \sum_{\substack{\BasisBeta:\BasisBeta\sim\mathbf{a} \\ \BasisBeta\sim\mathbf{a}'}}\delta_{s(\mathbf{a}),s(\mathbf{a}')} + \Or_n(1/\nu) \\
        &= 2^n \sum_{\substack{k: s(a_k)=1}} \sum_{\substack{\BasisBeta:\BasisBeta\sim\mathbf{a}}} \delta_{\mathbf{a},\mathbf{a}'} + \Or_n(1/\nu) \\
        &= |\mathbf{a}|2^n 3^{n-|\mathbf{a}|} \delta_{\mathbf{a},\mathbf{a}'} + \Or_n(1/\nu),
    \end{aligned}
\end{equation}
where $|\mathbf{a}|$ is the number of non-zero elements of $\mathbf{a}$, which appears in the above equation due to the summation over $k$. The factor $3^{n-|\mathbf{a}|}$ is the number of bases $\BasisBeta$ that are compatible with $\mathbf{a}$. From the above we can see that $J^\top J$ is almost a diagonal matrix with eigenvalue at least $2^n$ ($|\mathbf{a}|\geq 1$ because we omit the coefficient corresponding to the identity). We therefore reach the following lemma:
\begin{lem}
    \label{lem:singular_val_Jacobian}
    Let $J(\CoefMu,\nu)$ be the Jacobian matrix defined in \eqref{eq:jacobian_matrix}. Then its smallest singular value is lower bounded by $2^{n/2}-\Or_n(1/\nu)$.
\end{lem}

The limit of $J(\CoefMu,\nu)^\top J(\CoefMu,\nu)$ for $\nu\to\infty$ will be useful later on. Therefore we define
\begin{equation}
\label{eq:J0_defn}
    J_0 = \lim_{\nu\to \infty} J(\CoefMu,\nu)^\top J(\CoefMu,\nu)
\end{equation}
From \eqref{eq:JTJ_explicit} we can see that 
\begin{equation}
    \label{eq:J0_explicit}
    (J_0)_{\mathbf{a},\mathbf{a}'} = |\mathbf{a}|2^n 3^{n-|\mathbf{a}|} \delta_{\mathbf{a},\mathbf{a}'}.
\end{equation}

\subsection{The least-squares problem}
\label{sec:least_sqaures}

We will perform phase estimation experiments to obtain estimates for energy gaps collected into an array $\hat{\mathbf{E}}=(\hat{E}_{\Delta,k,\mathbf{s},\BasisBeta})$. With these estimates we will solve a least-squares problem to recover the Hamiltonian coefficients $\CoefLambda$. 
The exact values of the energy gaps then form a vector $\mathbf{E}^\star = \mathbf{E}(\CoefLambda,\nu)$.

A natural way to solve for $\CoefLambda$ is to solve a least-squares problem by minimizing the loss function
\begin{equation}
    \label{eq:loss_function}
    L(\hat{\mathbf{E}},\CoefMu,\nu) = \frac{1}{2}\|\hat{\mathbf{E}}-\mathbf{E}(\CoefMu,\nu)\|^2.
\end{equation}
The first-order optimality condition yields the equation
\begin{equation}
    \label{eq:first_order_optimality}
    \frac{\partial}{\partial \CoefMu} L(\hat{\mathbf{E}},\CoefMu,\nu) = (\mathbf{E}(\CoefMu,\nu)-\hat{\mathbf{E}})^\top \frac{\partial }{\partial\CoefMu}\mathbf{E}(\CoefMu,\nu)=0.
\end{equation}
By solving this equation around the exact value $\CoefLambda$ we will be able to obtain an estimate $g(\hat{\mathbf{E}})$ for the coefficients $\CoefMu$. We will first show that in a neighborhood of $\mathbf{E}^\star$, the equation \eqref{eq:first_order_optimality} leads to a unique $C^\infty$ map $g$.

\begin{lem}
    \label{lem:map_from_eigenvals_to_coefs}
    There exists $\nu_0= \Or_n(1)$ and $r=\Omega_n(1)$ such that when $\nu\geq \nu_0$, there exists a $C^\infty$ map $g: B_{r}(\mathbf{E}^\star)\to \RR^{4^n-1}$ such that
    \[
    \frac{\partial}{\partial \CoefMu} L(\hat{\mathbf{E}},g(\hat{\mathbf{E}}),\nu)=0.
    \]
    Moreover, its Jacobian matrix is
    \[
    \frac{\partial}{\partial\hat{\mathbf{E}}}g(\hat{\mathbf{E}}) = \left(\frac{\partial^2 L}{\partial \CoefMu^2}\right)^{-1}\frac{\partial \mathbf{E}}{\partial\CoefMu}.
    \]
\end{lem}

\begin{proof}
    We will use Theorem~\ref{thm:implicit_map} which is a modified implicit map theorem, to show the existence and uniqueness of such a function $g$. In the context of the equation \eqref{eq:first_order_optimality} we want to solve, the various objects in Theorem~\ref{thm:implicit_map} are:
    \[
    x=\hat{\mathbf{E}},\quad y = \CoefMu,\quad f(x,y) = \frac{\partial L}{\partial\CoefMu}(\hat{\mathbf{E}},\CoefMu).
    \]
    In the theorem we focus on a neighborhood of $(a,b)$ where $f(a,b)=0$. Here $a=\mathbf{E}^*$ and $b=\CoefLambda$, and $\frac{\partial L}{\partial\CoefMu}(\mathbf{E}^\star,\CoefLambda)=0$ because $\mathbf{E}^\star = \mathbf{E}(\CoefLambda)$ and \eqref{eq:first_order_optimality}.

    A key requirement in Theorem~\ref{thm:implicit_map} is that the matrix 
    \begin{equation}
    \label{eq:hessian}
        D_2 f(x,y) = \frac{\partial^2 L}{\partial\CoefMu^2}  =\left(\frac{\partial \mathbf{E}}{\partial\CoefMu}\right)^\top \left(\frac{\partial \mathbf{E}}{\partial\CoefMu}\right) + (\mathbf{E}(\CoefMu)-\hat{\mathbf{E}})^\top \frac{\partial^2\mathbf{E}}{\partial\CoefMu^2}
    \end{equation}
    satisfies \eqref{eq:D2f(x,y)_requirement}. Note that this matrix is the Hessian matrix of the objective function. For the first term on the right-hand side, by \eqref{eq:JTJ_explicit}, we have
    \begin{equation}
        \left\|\left(\frac{\partial \mathbf{E}}{\partial\CoefMu}\right)^\top \left(\frac{\partial \mathbf{E}}{\partial\CoefMu}\right)-J_0\right\| = \Or_n(1/\nu),
    \end{equation}
    when $|\lambda_{\mathbf{a}}|\leq 1$ for all $\mathbf{a}$. In fact we can replace this requirement with $|\lambda_{\mathbf{a}}|\leq 2$ for all $\mathbf{a}$ without affecting the validity of the statement. For the second term on the right-hand side, we will first deal with $\frac{\partial^2\mathbf{E}}{\partial\CoefMu^2}$. This is a rank-3 tensor and we will index its entries as $\frac{\partial^2 E_i}{\partial\mu_j\partial\mu_k}$. We define the following norm for a rank-3 tensor $A=(A_{ijk})$: for a vector $x$, we define $M_{jk} = \sum_i x_i A_{ijk}$, then
    \[
    \|A\|_T := \sup_{\|x\|\leq 1} \left\|M_i\right\|.
    \]
    With this definition, we have
    \[
    \left\|(\mathbf{E}(\CoefMu)-\hat{\mathbf{E}})^\top \frac{\partial^2\mathbf{E}}{\partial\CoefMu^2}\right\|\leq \|\mathbf{E}(\CoefMu)-\hat{\mathbf{E}}\|\left\|\frac{\partial^2\mathbf{E}}{\partial\CoefMu^2}\right\|_T.
    \]
    We can readily show that $\|A\|_T\leq \sum_{ijk}|A_{ijk}|$. By Lemma~\ref{lem:second_order_perturbation} we have
    \[
    \left|\frac{\partial^2 E_i}{\partial\mu_j\partial\mu_k}\right|=\Or_n(1/\nu).
    \]
    Therefore $\|\frac{\partial^2\mathbf{E}}{\partial\CoefMu^2}\|_T=\Or_n(1/\nu)$ and
    \[
    \left\|(\mathbf{E}(\CoefMu)-\hat{\mathbf{E}})^\top \frac{\partial^2\mathbf{E}}{\partial\CoefMu^2}\right\| = \Or_n(\|\mathbf{E}(\CoefMu)-\hat{\mathbf{E}}\|/\nu).
    \]
    Therefore
    \begin{equation}
        \label{eq:hessian_distance}
        \left\|\frac{\partial^2 L}{\partial\CoefMu^2}-J_0\right\|=\Or_n((1+\|\mathbf{E}(\CoefMu)-\hat{\mathbf{E}}\|)/\nu).
    \end{equation}
     From \eqref{eq:J0_explicit} we know that $\sigma_{\min}(J_0)=2^n$. Therefore we can guarantee \eqref{eq:D2f(x,y)_requirement} 
    with $\nu=\Omega_n(1)$ and $\|\mathbf{E}(\CoefMu)-\hat{\mathbf{E}}\|=\Or_n(\nu)$. The latter can be satisfied if 
    \[
    \|\CoefMu-\CoefLambda\|= \Or_n(\nu),\quad \|\hat{\mathbf{E}}-\mathbf{E}^\star\|=\Or_n(\nu),\quad |\lambda_{\mathbf{a}}|\leq 2,\ \forall \mathbf{a}.
    \]
    In the above we have used the mean value theorem and the fact that $\|\frac{\partial \mathbf{E}}{\partial\CoefMu}\|=\Or_n(1)$.
    Therefore, there exists $\nu_0=\Or_n(1)$ and $r_\lambda=\Omega_n(1)$, $r_E=\Omega_n(\nu)$, such that when $\nu\geq\nu_0$, $\|\CoefMu-\CoefLambda\|\leq r_\lambda$, and $\|\hat{\mathbf{E}}-\mathbf{E}^\star\|\leq r_E$, \eqref{eq:D2f(x,y)_requirement} holds,
    which in the context of our coefficient recovery task is
    \begin{equation}
    \label{eq:Hessian_approximated_by_J0}
        \left\|\frac{\partial^2 L}{\partial\CoefMu^2}(\hat{\mathbf{E}},\CoefMu,\nu)-J_0\right\|\leq \frac{1}{2}\sigma_{\min}(J_0).
    \end{equation}
    
    The only remaining component in Theorem~\ref{thm:implicit_map} we now need is then
    \[
    M_1 =\sup_{(x,y)\in B_{r_x}(a)\times B_{r_y}(b)}\|D_1 f(x,y)\|=\max_{\substack{\|\CoefMu-\CoefLambda\|\leq r_\lambda \\\|\hat{\mathbf{E}}-\mathbf{E}^\star\|\leq r_E}} \left\|\frac{\partial^2 L}{\partial\CoefMu\partial \hat{\mathbf{E}}}\right\| = \max_{\substack{\|\CoefMu-\CoefLambda\|\leq r_\lambda \\\|\hat{\mathbf{E}}-\mathbf{E}^\star\|\leq r_E}} \left\|\frac{\partial \mathbf{E}}{\partial\CoefMu}\right\|=\Or_n(1). 
    \]
    Therefore we have all the components needed for applying Theorem~\ref{thm:implicit_map}, and direct calculation shows that the radius of the ball in which the implicit function $g$ can be uniquely defined is
    \[
    r = \min\{r_\lambda/(2\|J_0^{-1}\|M_1),r_E\} = \Omega_n(1).
    \]
\end{proof}

In the above we have also proved that the optimization problem defined in \eqref{eq:loss_function} is locally strongly convex:
\begin{prop}[Strong convexity]
    \label{prop:strong_convexity}
    There exists $\nu_0=\Or_n(1)$ and $r_\lambda=\Omega_n(1)$, $r_E=\Omega_n(\nu)$, such that when $\nu\geq\nu_0$, $\|\CoefMu-\CoefLambda\|\leq r_\lambda$,
    the objective function $L(\hat{E},\CoefMu,\nu)$ in \eqref{eq:loss_function} satisfies
    \[
    \lambda_{\min}\left(\frac{\partial^2 L}{\partial \CoefMu^2}(\hat{E},\CoefMu,\nu)\right)\geq 2^{n-1},
    \]
    where $\lambda_{\min}(A)$ denotes the smallest eigenvalue of a real symmetric matrix $A$.
\end{prop}

\begin{proof}
    This immediately follows from \eqref{eq:Hessian_approximated_by_J0} and the explicit form of $J_0$ given in \eqref{eq:J0_explicit}.
\end{proof}

\begin{thm}
    \label{thm:recover_coefs}
    There exists $\nu_0= \Or_n(1)$ and $r=\Omega_n(1)$ such that when $\nu\geq \nu_0$, there exists a $C^\infty$ map $g: B_{r}(\mathbf{E}^\star)\to \RR^{4^n-1}$ such that
    \[
    \|g(\hat{\mathbf{E}})-\CoefLambda\|\leq C\|\hat{\mathbf{E}}-\mathbf{E}^\star\|,
    \]
    for all $\hat{\mathbf{E}}\in B_{r}(\mathbf{E}^\star)$,
    where $C$ is a constant that depends only on $n$.
    Moreover, there exists $r_\lambda=\Omega_n(1)$, such that for any $\hat{\mathbf{E}}\in B_{r}(\mathbf{E}^\star)$, $g(\hat{\mathbf{E}})$ can be computed to precision $\epsilon_1$ with $\Or_n(\log(1/\epsilon_1))$ queries to the function $\mathbf{E}(\CoefMu,\nu)$, when given an initial guess $\CoefLambda^0$ satisfying $\|\CoefLambda^0-\CoefLambda\|< r_\lambda$.
\end{thm}

\begin{proof}
    The map $g$ is readily provided in Lemma~\ref{lem:map_from_eigenvals_to_coefs}. The deviation of $g(\hat{\mathbf{E}})$ from $\CoefLambda$ can be bounded using the expression for $\partial_{\hat{\mathbf{E}}}g(\hat{\mathbf{E}})$ in Lemma~\ref{lem:map_from_eigenvals_to_coefs}, the eigenvalue lower bound for $\partial^2 L/\partial\CoefMu^2$ in Proposition~\ref{prop:strong_convexity}, and the fact that $\|\partial\mathbf{E}/\partial \CoefMu\|=\Or_n(1)$. The algorithm for computing $g(\hat{\mathbf{E}})$ is an iterative algorithm derived from the existence proof of $g$. It is described in pseudocode in Algorithm~\ref{alg:nt_method} and its convergence is proved in Lemma~\ref{lem:convergence_of_iterative_alg}.
\end{proof}
 

\section{Eigenstate and eigenvalue perturbations}

\begin{lem}
    \label{lem:eigenstate_perturbation_bound}
    Let $H_0$ be any Hermitian matrix. let $\ket{\Phi_k}$ be its $k$th lowest eigenstate corresponding to an eigenvalue $E_k$, which is separated from the rest of the spectrum by a gap $\Delta$. Assume a Hermitian matrix $V$ satisfies $\|V\|\leq \Delta/3$, and we let $\ket{\Psi_k}$ be the $k$th lowest eigenstate of $H_0+V$. Then there exists $\theta\in[0,2\pi)$ such that
    \[
     \left\|\ket{\Psi_k}-e^{i\theta}\ket{\Phi_k}\right\|\leq 3\|V\|/\Delta.
    \]
\end{lem}

\begin{proof}
    This lemma is similar to Lemma~3.1 in \cite{BravyiDiVincenzoLoss2011schrieffer}, but we provide a proof here for completeness.
    We let $H(s) = H_0 + sV$. Then $H(0)=H_0$ and $H(1)=H_0+V$. By the stability fo eigenvalues which follows from Weyl's inequality \cite{Tao2010254a}, the $k$th eigenvalue of $H(s)$ is separated from the rest of the spectrum by a gap of at least $\Delta-2\|V\|\geq \Delta/3$ for all $s\in[0,1]$. We can then uniquely define the $k$th eigenstate to be $\ket{\Psi_k(s)}$ Then by the eigenstate perturbation theory \cite{Bha97:Matrix-Analysis}, the phases of these states can be chosen in such a way that
    \[
    \frac{\dd}{\dd s}\ket{\Psi_k(s)} = -(H(s)-E_k(s))^{+}V\ket{\Psi_k(s)},
    \]
    where $(H(s)-E_k(s))^{+}$ denotes the pseudoinverse of $H(s)-E_k(s)$, which satisfies $\|(H(s)-E_k(s))^{+}\|\leq 3/\Delta$ due to the presence of the energy gap. Therefore we have
    \[
    \left\|\frac{\dd}{\dd s}\ket{\Psi_k(s)}\right\|\leq 3\|V\|/\Delta.
    \]
    Since $\ket{\Phi_k}=e^{i\theta_0}\ket{\Psi_k(0)}$ and $\ket{\Psi_k}=e^{i\theta_1}\ket{\Psi_k(1)}$, for some $\theta_0,\theta_1\in \RR$, from the above we can see that
    \[
    \left\|\ket{\Psi_k}-e^{i\theta}\ket{\Phi_k}\right\|=\left\|\int_0^1 \frac{\dd}{\dd s}\ket{\Psi_k(s)} \dd s\right\|\leq 3\|V\|/\Delta,
    \]
    with an appropriately chosen $\theta\in[0,2\pi)$.
\end{proof}

Direct calculation based on eigenvalue and eigenstate perturbation theory also yields the following lemma:
\begin{lem}
    \label{lem:second_order_perturbation}
    Let $H(\lambda_1,\lambda_2)=H_0 + \lambda_1 V_1 + \lambda_2 V_2$. Let $E_k(\lambda_1,\lambda_2)$ be the $k$th smallest eigenvalue of $H(\lambda_1,\lambda_2)$, and is separated from the rest of the spectrum by a positive gap when $\lambda_1=\lambda_2=0$. Let $\ket{\Psi_k(\lambda_1,\lambda_2)}$ be the corresponding eigenstate, which is uniquely defined up to a global phase in a neighborhood of  $\lambda_1=\lambda_2=0$. Then
    \[
    \frac{\partial^2 E_k}{\partial \lambda_i\partial\lambda_j} = 2\Re \braket{\Psi_k|V_j(E_k-H)^{+}V_i|\Psi_k}.
    \]
\end{lem}


\section{The implicit map theorem}

\begin{thm}[Modified implicit map theorem]
\label{thm:implicit_map}
Let $U$, $V$ be open sets in Banach spaces $E$, $F$ respectively, and let
$f: U \times V \rightarrow G$
be a $C^p$ map. Let $(a, b) \in U \times V$, and we assume.
$f(a, b) = 0$. 
We also assume
\begin{align}
\label{eq:D2f(x,y)_requirement}
    \|D_2f(x,y)-J_0 \|_2\leq \sigma_{min}(J_0)/2, \, \forall x\in B_{r_x}(a), \, y\in B_{r_y}(b),
\end{align} 
for some fixed $r_x, r_y$, and toplinear \footnote{Here ``toplinear'' means the linear operator is bounded and has bounded inverse.} $J_0$.
Then for $r= \min\{r_y/(2\|J_0^{-1}\|M_1),r_x\}$, where 
\[
M_1=\sup_{(x,y)\in B_{r_x}(a)\times B_{r_y}(b)}\|D_1 f(x,y)\|,
\]
there exists a continuous map $g: B_{r}(a) \rightarrow V$ with $g(a) = b$, such that
$f(x, g(x)) = 0$
for all $x \in B_r(a)$. $g$ is uniquely determined in $B_r(a)$, and is also of class $C^p$.
\end{thm}

\begin{proof}

Fixing $x \in B_{r_x}(a)$, we want to find $y$ near $b$ such that $f(x, y) = 0$. We define the map $\Phi_x: B_{r_y}(b)\to F$ to be
\begin{align}
\Phi_x(y) := y - J_0^{-1} f(x, y).
\end{align}
Then the fixed points of $\Phi_x$ correspond to solutions of $f(x, y) = 0$. 

We first show that $\Phi_x(y)$ is contractive in $B_{r_y}(b)$. For any $y_1, y_2 \in B_{r_y}(b)$:
\begin{align}
\Phi_x(y_1) - \Phi_x(y_2) &= y_1 - y_2 - J_0^{-1}\left(f(x, y_1) - f(x, y_2)\right) \notag\\
&= y_1-y_2 -J_0^{-1}\int_0^1D_2f(x,y_2+t(y_1-y_2))(y_1-y_2)\,\dd t, \notag\\
&=\left( I - J_0^{-1} \int_0^1 D_2 f(x, y_2 + t(y_1 - y_2)) \, \dd t  \right)(y_1 - y_2)
\end{align}
Therefore
\begin{align}
\label{eq:contraction}
\|\Phi_x(y_1) - \Phi_x(y_2)\|
&= \left\| \left(I - J_0^{-1} \int_0^1 D_2 f(x, y_2 + t(y_1 - y_2)), \dd t \right)(y_1 - y_2) \right\| \notag\\
&\leq \left\| I - J_0^{-1} \int_0^1 D_2 f(x, y_2 + t(y_1 - y_2))\, \dd t \right\| \cdot \|y_1 - y_2\|\notag\\
&= \|J_0^{-1}\|\cdot\left\|J_0-\int_0^1D_2f(x,y_2 + t(y_1 - y_2))\,\dd t\right\|\cdot\|y_1-y_2\| \notag\\
&\leq \|J_0^{-1}\|\cdot \frac{\sigma_{\min}(J_0)}{2}\cdot\|y_1-y_2\| \notag\\
&= \frac{1}{2}\|y_1-y_2\|.
\end{align}

Next we show that $\Phi_x(y)$ maps $B_{r_y'}(b)$ to itself when $x\in B_{r'_x}(a)$.
By the triangle inequality
\begin{align}\label{eq: Phi_x_minus_b}
    \|\Phi_x(y)-b\|\leq \|\Phi_x(y)-\Phi_x(b)\|+\|\Phi_x(b)-b\|.
\end{align}
Using \eqref{eq:contraction} we can bound the first term on the RHS of \Cref{eq: Phi_x_minus_b}:
\begin{align}
    \|\Phi_x(y)-\Phi_x(b)\| &\leq  \frac{1}{2}\|y-b\|, \label{eq: phixy_phixb}
\end{align}
%
The second term on the RHS of \Cref{eq: Phi_x_minus_b} yields
\begin{align}
    \|\Phi_x(b)-b\| &= \|J_0^{-1}f(x,b)\| \tag{by definition of $\Phi_x$}\\
    &\leq \|x-a\|\int_0^1 \|J_0^{-1}D_1f(tx+(1-t)a,b)\|\,\dd t \notag\\
    &\leq r'_x \int_0^1 \|J_0^{-1}D_1f(tx+(1-t)a,b)\|\,\dd t \tag{since $x\in B_{r'_x}(a)$} \\
    &\leq r'_x \|J^{-1}_0\|M_1, \label{eq: rxJM1}
\end{align}
where the last inequality used the definition of $M_1$.
Setting $r'_x = \min\{r_y/(2\|J_0^{-1}\|M_1),r_x\}$ and $r'_y=2\|J_0^{-1}\|M_1 r_x'$ as done in the statement of the theorem ensures that $r'_y\leq r_y$ and $\|\Phi_x(b)-b\|\leq r_y/2$. Inserting \Cref{eq: phixy_phixb,eq: rxJM1} into \Cref{eq: Phi_x_minus_b} yields
\begin{align} 
   \| \Phi_x(y)-b\|&\leq \frac{1}{2}\|y-b\|+\|J_0^{-1}\|M_1 r_x' \notag\\
   &\leq \frac{1}{2}r'_y + \frac{1}{2}r'_y \tag{since $y\in B_{r'_y}(b)$ and $r'_y/2=\|J_0^{-1}\|M_1 r_x'$}\\
   &=r'_y.
\end{align}
Therefore $\Phi_x(y)$ maps $B_{r_y'}(b)$ to itself.

Using the Shrinking Lemma~\cite{lang2012real}, we then know that for each $x \in B_{r'_x}(a)$, there exists a unique fixed point $g(x) \in B_{r'_y}(b)$ such that $\Phi_x(g(x)) = g(x)$ which implies that $f(x, g(x)) = 0$.
This defines a function $g: B_{r'_x}(a) \to B_{r'_y}(b)\subseteq V$ with $g(a) = b$ and establishes its uniqueness.




We last show that $g$ is of class $C^p$. 
First we show that $g$ is differentiable. 
Let $x \in B_{r_x}(a)$, and let $x_1 \in B_{r_x}(a)$ be fixed. Since $f \in C^1$, we can write the first-order Taylor expansion of $f$ around the point $(x_1, g(x_1))$:
\begin{align}
f(x, g(x)) 
&= f(x_1, g(x_1)) 
+ D_1 f(x_1, g(x_1))(x - x_1) \nonumber \\
&\quad + D_2 f(x_1, g(x_1))(g(x) - g(x_1)) + r(x),
\end{align}
where the remainder term satisfies $\|r(x)\| \leq \varepsilon_1 \|x - x_1\| + \varepsilon_2 \|g(x) - g(x_1)\|$, for some $\varepsilon_1, \varepsilon_2 \to 0$ as $x \to x_1$. Since $f(x, g(x)) = 0 = f(x_1, g(x_1))$, the above expansion tells us
\begin{align*}
0 = D_1 f(x_1, g(x_1))(x - x_1) 
+ D_2 f(x_1, g(x_1))(g(x) - g(x_1)) + r(x).
\end{align*}
Now solve for $g(x) - g(x_1)$:
\begin{align*}
g(x) - g(x_1) 
&= -[D_2 f(x_1, g(x_1))]^{-1} D_1 f(x_1, g(x_1))(x - x_1) + R(x) \notag\\
&:= A(x-x_1)+R(x),
\end{align*}
where
$R(x) := -[D_2 f(x_1, g(x_1))]^{-1} r(x)$, $A := -[D_2 f(x_1, g(x_1))]^{-1} D_1 f(x_1, g(x_1))$.
We estimate the remainder using the bound on $r(x)$:
\begin{align*}
\|R(x)\| \leq C_1 \varepsilon_1 \|x - x_1\| + C_1 \varepsilon_2 \|g(x) - g(x_1)\|,
\end{align*}
where $C_1 = \|[D_2 f(x_1, g(x_1))]^{-1}\|$.
Now observe that
\begin{align*}
\|g(x) - g(x_1)\| 
&\leq \|A(x - x_1)\| + \|R(x)\| \nonumber \\
&\leq \|A\| \cdot \|x - x_1\| + C_1 \varepsilon_1 \|x - x_1\| + C_1 \varepsilon_2 \|g(x) - g(x_1)\|.
\end{align*}
Rewriting we get:
\begin{align}
(1 - C_1 \varepsilon_2)\|g(x) - g(x_1)\| 
\leq (\|A\| + C_1 \varepsilon_1) \cdot \|x - x_1\|.
\end{align}
For sufficiently small $x$ near $x_1$, we can choose $\varepsilon_2$ small enough that $1 - C_1 \varepsilon_2 > \frac{1}{2}$, so we obtain:
\begin{align}
\|g(x) - g(x_1)\| = O(\|x - x_1\|).
\end{align}
Substituting this back into the estimate for $\|R(x)\|$:
\begin{align}
\|R(x)\| \leq C_1 \varepsilon_1 \|x - x_1\| + C_1 \varepsilon_2 O(\|x - x_1\|) = o(\|x - x_1\|),
\end{align}
since both $\varepsilon_1$ and $\varepsilon_2$ tend to $0$ as $x \to x_1$.
Therefore, we conclude:
\begin{align}
g(x) - g(x_1) = A(x - x_1) + o(\|x - x_1\|),
\end{align}
so $g$ is differentiable at $x_1$, with $Dg(x_1) = A$.



Then we can show that $g$ is of class $C^p$. By definition of $\Phi_x(y)$ and $g(x)$ being its unique fixed point, we have
\begin{align}
f(x, g(x)) = 0 \quad \forall x \in B_{r'_x}(a).
\end{align}
Since $f$ is of class $C^p$, we can differentiate this identity:
\begin{align}
D_1 f(x, g(x)) + D_2 f(x, g(x)) \cdot Dg(x) = 0. \label{eq:implicit_diff}
\end{align}
Since $D_2 f(x, g(x))$ is invertible for all $x \in B_{r_x}(a)$ by assumption, we can solve for $Dg(x)$:
\begin{align}
Dg(x) = -\left(D_2 f(x, g(x))\right)^{-1} \cdot D_1 f(x, g(x)). \label{eq:dg_formula}
\end{align}
This shows that $Dg(x)$ is a composition of $C^{p-1}$ functions, so $Dg(x)$ is of class $C^{p-1}$ which implies that $g$ is of class $C^p$.
Setting $r=r'_x$ completes the proof.
\end{proof}

\begin{algorithm}
\caption{Iterative method for Solving $f(\hat{x}, y) = 0$}
\label{alg:nt_method}

\KwData{
    Smooth map $f : U \times V \to G$,  input $\hat{x} \in U$, invertible linear map $J_0 \approx D_2 f(\hat{x}, y)$, initial guess $y_0 \in V$, tolerance $\varepsilon > 0$.
}
\KwResult{
    Approximate solution $\hat{y}$ such that $f(\hat{x}, \hat{y}) \approx 0$.
}
Define the fixed-point map as in \Cref{thm:implicit_map}: $\Phi(y) := y - J_0^{-1} f(\hat{x}, y)$.

Initialize $y \gets y_0$.

\Repeat{$\delta < \varepsilon$}{
    $y_{\text{new}} \gets \Phi(y) = y - J_0^{-1} f(\hat{x}, y)$.\\
    $\delta \gets \|y_{\text{new}} - y\|$.\\
    $y \gets y_{\text{new}}$.
}

\Return $\hat{y} \gets y$.

\end{algorithm}

The existence proof in \Cref{thm:implicit_map} can be turned into an algorithm described in Algorithm~\ref{alg:nt_method}.
Next, we analyze its convergence:

\begin{lem}
\label{lem:convergence_of_iterative_alg}
    Using the definitions in \Cref{thm:implicit_map} and let the fixed-point iteration be the map used in \Cref{alg:nt_method}:
\begin{align}
\Phi(y) := y - J_0^{-1} f(\hat{x}, y), \quad y_{k+1} := \Phi(y_k),
\end{align}
starting from an initial point $y_0 \in B_{r_y}(b)$. 
Then $\Phi$ is a contraction on $B_{r_y}(b)$ with contraction constant $1/2$. Consequently, the iteration converges to the unique fixed point $\hat{y}$ satisfying $f(\hat{x}, \hat{y}) = 0$.
In particular, to ensure $\|y_k - \hat{y}\| \leq \varepsilon$, it suffices to take iteration
\begin{align}
k \geq  \log_{2} \left( \frac{\|y_0 - \hat{y}\|}{\varepsilon} \right).
\end{align}
\end{lem}

\begin{proof}
By \Cref{thm:implicit_map}, $\Phi$ is contractive on $B_{r_y}(b)$ with contraction constant $q=1/2<1$.
By the Shrinking Lemma, the iteration $y_{k+1} = \Phi(y_k)$ converges to the unique fixed point $\hat{y}$ such that $\Phi(\hat{y}) = \hat{y}$ with $f(\hat{x}, \hat{y}) = 0$, and $\|y_k - \hat{y}\| \leq q^k \|y_0 - \hat{y}\|$.
To ensure $\|y_k - \hat{y}\| \leq \varepsilon$, solve $q^k \|y_0 - \hat{y}\| \leq \varepsilon$ yields $k \geq \log_{1/q} \left( \frac{\|y_0 - \hat{y}\|}{\varepsilon} \right)=\log_{2} \left( \frac{\|y_0 - \hat{y}\|}{\varepsilon} \right)$.
\end{proof}

\section{Numerical experiments}
\label{sec:numerical}





\noindent
\textbf{The Hamiltonians.}
In Figure~2 in the main text, the Hamiltonians are chosen as follows.
For the single-qubit Hamiltonian $H_{\text{true}}=\lambda_1\sigma^1 +\lambda_2\sigma^2 +\lambda_3\sigma^3$, we choose $\lambda_1=0.1$, $\lambda_2=0.5$, and $\lambda_3=0.3$. 
For the two-qubit case, we parametrize the Hamiltonian in the Pauli basis as
\begin{equation}
H(\boldsymbol{\lambda})
=
\sum_{(a_1,a_2)\in\{0,1,2,3\}^2\setminus (0,0)}
\lambda_{(a_1,a_2)} \,
\sigma^{a_1}\otimes \sigma^{a_2},
\end{equation}
and take $\boldsymbol{\lambda} = (0.1,0.2,0.3,0.5,0.6,0.3,0.2,0.1,0.1,\,
0.2,0.1,0.1,0.3,0.22,0.15)$.

\vspace{1em}
\noindent
\textbf{Numerical optimization.}
The initial guesses used in the optimization procedure in Figure~2 are chosen as follows. 
In the one-qubit case we take the initial guess $\boldsymbol{\lambda}_{\text{guess}}=(0.09, 0.51, 0.29)$, and in the two-qubit case we have 
$\boldsymbol{\lambda}_{\text{guess}}=(0.11,0.21,0.32, 0.51,0.63,0.31,0.22,0.11,0.11, 0.22,0.11,0.11, 0.33,0.22,0.15)$.
It is worth noting that in the experiments, even when the initial guess is much farther from the true coefficients, the algorithm can still converge and recover the correct coefficients. As shown in \Cref{fig:initial_guess}, the final error achieved by our protocol does not significantly change even if the initial guess is selected to be very far from the actual solution. 

Numerical optimization is performed using the least-squares solver in \texttt{Scipy}. We first obtain estimates of the spectral gaps collected into a vector $\mathbf{E}^\star$ through phase estimation experiments described in Section~\ref{sec:the_phase_est_experiment}, and then fit these estimates by optimizing the Hamiltonian parameters.

\begin{figure}
\centering
\includegraphics[width=0.6\textwidth]{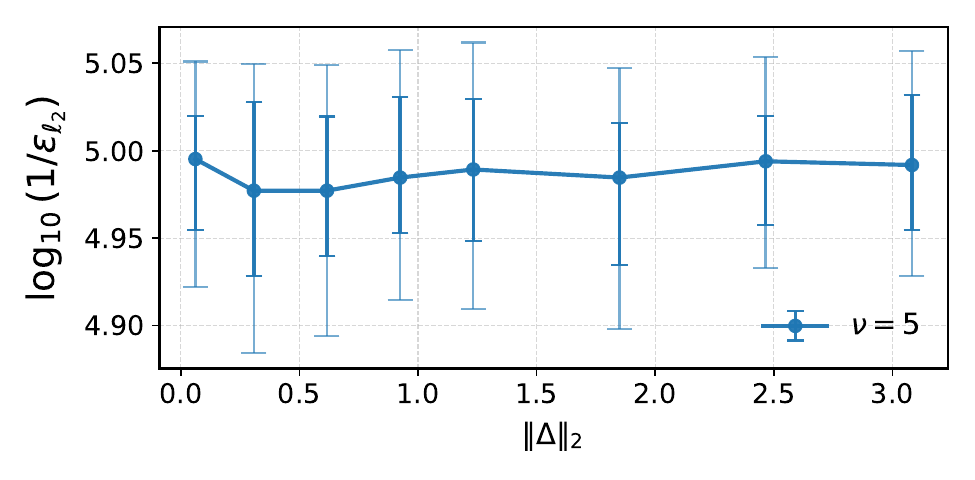}
\caption{Two-qubit Hamiltonian learning with different initial guesses. For fixed $\nu=5$ and total evolution time $1.4\times 10^8$ (corresponding to a target precision of $10^{-4}$ for energy gaps), we vary the initial guess and estimate the final accuracy.
The $x$-axis is $\|\Delta\|$, which is the $\ell^2$ distance between the initial guess $\CoefLambda_{\text{guess}}$ and the true coefficient vector $\CoefLambda$ (here $\Delta=\CoefLambda_{\text{guess}}-\CoefLambda$). 
}
\label{fig:initial_guess}
\end{figure}

\vspace{1em}
\noindent \textbf{SPAM noise.} The SPAM noise is modeled through a bit-flip channel applied before measurement. Concretely, the measurement result flips from $0$ to $1$ or $1$ to $0$ with probability $\eta$ independently for each qubit. In Figure~2 in the main text we set $\eta=0.05$ for the single-qubit case and $\eta=0.03$ for the two-qubit case. In Figure~\ref{fig:increasing_error} we gradually increase $\eta$ to $0.25$ for the single-qubit case and compute the final error. We observe that the Heisenberg-limited scaling is preserved for $\eta$ up to about $0.2$.

\begin{figure}
\centering
\includegraphics[width=0.4\textwidth]{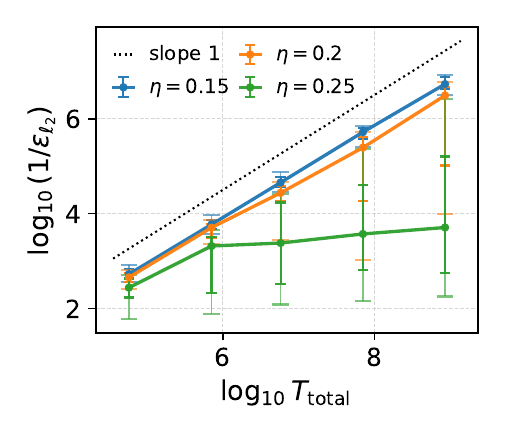}
\caption{One-qubit experiment with increasing SPAM error rate $\eta$ for fixed $\nu=3$. 
}
\label{fig:increasing_error}
\end{figure}

\section{The total evolution time lower bound}

An algorithm that can learn a Hamiltonian as described in the original Theorem~1 can be used to solve a simpler learning task of discrimination between the following single-qubit Hamiltonians. Assume we are given either $Z$ or $Z+\epsilon X$ (note that in Theorem~2 in the main text we use a normalized version of $Z+\epsilon X$, but here we use the unnormalized version first and then extend the result to the normalized case). If there is a learning
algorithm with a total evolution time at most $T$ that succeeds in the learning task with probability at least $q$,
then we can use the learning algorithm to successfully distinguish between $Z$ and the normalized version of $Z+\epsilon X$ with  probability at least $q$. We derive a lower-bound on the total Hamiltonian evolution time needed for this simple discrimination task, which we find to be $\in \Omega(\frac{1}{4\nu\epsilon+3\epsilon^2})$ in Theorem \ref{thm:LB} in the following text. This implies a lower-bound on $T$ for the original learning task. \\

The proof is split into 4 subsections. In the first subsection, we find an upper-bound for the diamond distance between the action of the unitaries generated by the Hamiltonians $Z$ and $Z+\epsilon X$. Next we use this to derive an upper bound for the total variation (TV) distance between the distribution over measurement outcomes under a single experiment. Then we upper bound the TV distance between distributions of measurement outcomes under many experiments which can be adaptively chosen based on previous outcomes. We leverage a learning tree representation to capture the adaptivity introduced. Finally, we utilize Le Cam's two-point method to turn this upper bound into a lower bound for the total evolution time, as shown in in Theorem \ref{thm:LB}. 

\subsection{Diamond distance}

\begin{lem} 
We have the following bound for the distance in terms of the spectral norm  between
$\exp\big(-it(Z+\nu P + \epsilon X)\big)$ and $\exp\big(-it(Z+\nu P)\big)$, where $\epsilon \ge 0$, $P$ is a traceless Hermitian operator with eigenvalues $\pm 1$ denoting the direction of the external magnetic field, $0 \le \nu \le 1$ denoting the strength of the magnetic field and $t\ge 0$ denotes how long the Hamiltonian is applied:
    \begin{align}
    &\| \exp\big(-it(Z+\nu P + \epsilon X)\big)-\exp\big(-it(Z+\nu P)\big)\|  \le 2\epsilon(2\nu + \epsilon)t+ (2\epsilon(2\nu + \epsilon)+\epsilon  )\min( t,\frac{1}{1-\nu})
\end{align}

\begin{proof}
We define $w(\epsilon)=\|Z+\nu P + \epsilon X\|$. 
    \begin{align*}
    &\| \exp\big(-it(Z+\nu P + \epsilon X)\big)-\exp\big(-it(Z+\nu P)\big)\| \\
    &= \Bigl\|(\cos(w(\epsilon)t) - \cos(w(0)t))\mathbb{I} -i\sin(w(\epsilon)t)\,\frac{Z+\nu P + \epsilon X}{w(\epsilon)}+i \sin(w(0)t)\,\frac{Z+\nu P}{w(0)}
          \Bigr\|
\end{align*}
The last expression follows from the following identity for single-qubit operators: $\exp\!\big(-i t\, \vec{a}\!\cdot\!\vec{\sigma}\big)        
    = \cos(|\vec{a}|t)\mathbb{I} 
        -         i\sin(|\vec{a}|t)\,\frac{\vec{a}}{|\vec{a}|}\!\cdot\!\vec{\sigma}$.
Therefore
\begin{align*}
    &\| \exp\big(-it(Z+\nu P + \epsilon X)\big)-\exp\big(-it(Z+\nu P)\big)\| \\
    &\le \bigl|\cos(w(\epsilon)t) - \cos(w(0)t)\bigr|
      + \Bigl\|
            \sin(w(\epsilon)t)\,\frac{Z+\nu P  +  \epsilon X}{w(\epsilon)}
            - \sin(w(0)t)\,\frac{Z+\nu P}{w(0)}
          \Bigr\|\nonumber\\
    &\le \bigl|\cos(w(\epsilon)t) - \cos(w(0)t)\bigr|
      + \Bigl\|
            \big(\sin(w(\epsilon)t)\frac{w(0)}{w(\epsilon)}
            - \sin(w(0)t)\big)\frac{Z+\nu P}{w(0)}
          \Bigr\|
      + \frac{\epsilon}{w(\epsilon)}\|\sin(w(\epsilon) t)X\|\nonumber\\
    &= \bigl|\cos(w(\epsilon)t) - \cos(w(0)t)\bigr|
      + \Bigl|
            \sin(w(\epsilon)t)\frac{w(0)}{w(\epsilon)}
            - \sin(w(0)t)
          \Bigr|
      + \frac{\epsilon}{w(\epsilon)}|\sin(w(\epsilon) t)|\nonumber\\      
    \end{align*}

    The last line follows because $\|Z+\nu P\|=w(0).$ Then,
    \begin{align}
    &\| \exp\big(-it(Z+\nu P + \epsilon X)\big)-\exp\big(-it(Z+\nu P)\big)\| \\
    &\le \bigl|\cos(w(\epsilon)t) - \cos(w(0)t)\bigr|
      + \Bigl|
            \sin(w(\epsilon)t)
            - \sin(w(0)t)
          \Bigr|\nonumber\\
    &\qquad{}+ \Bigl| \frac{w(0)}{w(\epsilon)}-1 \Bigr|\cdot|\sin(w(\epsilon) t)|
      + \frac{\epsilon}{w(\epsilon)}|\sin(w(\epsilon) t)|\nonumber\\
    &\le 2|w(\epsilon)-w(0)|t
      + \Bigl(\frac{|w(\epsilon)-w(0)|}{w(\epsilon)}+\frac{\epsilon}{w(\epsilon)}\Bigr)
        |\sin(w(\epsilon) t)| \label{eq:halfwayLB}
\end{align}

The last inequality follows from $|\sin(b)-\sin(a)|\le |b-a|$ and $|\cos(b)-\cos(a)|\le |b-a|$.

When $\nu < 1$, we argue that $w(\epsilon) \ge 1-\nu$ using the triangle inequality. 
\begin{align*}
    &w(\epsilon)=\|Z+\nu P+ \epsilon X\| \ge \|Z+\epsilon X\| -\|\nu P\| \\
    &=\sqrt{1+\epsilon^2}-\nu \ge 1-\nu.
\end{align*}

In Section~\ref{appendix:difference}, we will show that $|w(\epsilon)-w(0)| \le \epsilon(2\nu + \epsilon)$ when $\nu < 1$. We insert this in Eq. \eqref{eq:halfwayLB}.
\begin{align*}
    & \le 2|w(\epsilon)-w(0)|t+ (\frac{|w(\epsilon)-w(0)|}{w(\epsilon)}+\frac{\epsilon}{w(\epsilon)}  )|\sin(w(\epsilon) t)|\\
    & \le 2\epsilon(2\nu + \epsilon) t+ (\epsilon(2\nu + \epsilon)+\epsilon  )\min( t,\frac{1}{w(\epsilon)})\\
     & \le 2\epsilon(2\nu + \epsilon) t+ \epsilon(2\nu + \epsilon+1) \min( t,\frac{1}{1-\nu}).
\end{align*}
The second inequality follows because  $|\sin(w(\epsilon) t)|\le 1$, and 
    $\frac{|\sin(w(\epsilon) t)|}{w(\epsilon)} \le \frac{w(\epsilon) t}{w(\epsilon)}= t$.

\end{proof}
\end{lem}
\subsection{TV upper bound for a single experiment}
A single experiment is defined the following way. Given an unknown Hamiltonian \(H\) equal to either \(Z\) or \(Z+\epsilon X\), a single experiment is \(E\)  specified by the following parameters,
\begin{enumerate}
  \item an arbitrary \(N\)-qubit initial state \(\ket{\psi_0}\in\mathbb{C}^{2^{N}}\) with an integer \(N\ge 1\),
  \item an \(N\)-qubit unitary of the following form,  
\begin{equation*}
U_{L+1}(\exp\big(-it_L(H+\nu_L P_L)\big)\otimes I\bigr)U_L\cdots U_3\bigl(\exp\big(-it_2(H+\nu_2 P_2)\big)\otimes I\bigr)U_2(\exp\big(-it_1(H+\nu_1 P_1)\big)\otimes I\bigr)U_1,
\end{equation*}
where $P_i$ and $0\le \nu_i \le \nu$ are arbitrary single-qubit traceless hermitian operators with eigenvalues $\pm 1$ representing direction of the external magnetic field and the strength of the field respectively, 

\item an arbitrary POVM \(\mathcal{F}=\{M_i\}_i\) on the \(N\)-qubit system,
\end{enumerate}
for some arbitrary integer \(L\) which represents the number of times the Hamiltonian is evolved, arbitrary evolution times \(t_1,\dots,t_L\) and arbitrary \(N\)-qubit unitaries \(U_1,\dots,U_L,U_{L+1}\). Here \(I\) is the identity unitary on the \(N-1\) qubits. Here $L+1$ can be interpreted as the number of discrete control operations involved in a single experiment. 
\begin{thm}\label{Thm:singl-experimentLB}
    Consider the task of deciding whether a given single-qubit Hamiltonian $H$ is $Z$ or $Z+\epsilon X$ 
    where the experimentalist can apply a magnetic field of arbitrary strength at most $\nu<1$ along a direction of choice denoted by a single-qubit Hermitian $P$ with with eigenvalues $\pm 1$. A single experiment where the Hamiltonian is evolved $L$ times with a total evolution time $T$ has a total variation distance bounded by 
    \[
    TV(p_{\epsilon},p_0) \le (2\nu\epsilon+\epsilon^2)T +\frac{\epsilon(2\nu +\epsilon+1)}{2(1-\nu)}L.
    \]
\end{thm} 
\begin{proof}
We define $\ket{\psi_{\epsilon}} = U_L U_{\epsilon}(t_L) \cdots U_3 U_{\epsilon}(t_2) U_2 U_{\epsilon}(t_1) U_1 \ket{\psi_0}$, where $U_{\epsilon}(t_l)=\exp\big(-it_l(H+\nu_l P_l)\big)\otimes I, \forall 1\le l \le L$. 
By triangle inequality and telescoping sum, we have the following upper bound on the trace distance,
\begin{equation}
\|\ket{\psi_\epsilon}\!\bra{\psi_\epsilon} - \ket{\psi_0}\!\bra{\psi_0}\|_1
\le 
\biggl\| \prod_i U_i(\exp\bigl(-it_i(Z+\nu_i P_i+\epsilon X)\bigr)\otimes I) 
    - \prod_i U_i(\exp\bigl(-it_i(Z+\nu_i P_i)\bigr)\otimes I)  \biggr\|.
\end{equation}

We can now upper bound the total variation distance
for the classical probability distribution when we measure the final state using the ideal POVM measurement 
$\mathcal{F} = \{ M_i \}_i$. By definition, we have $p_{\epsilon}(i) = \bra{\psi_{\epsilon}} M_i \ket{\psi_{\epsilon}}$. 
\begin{align}
&TV(p_{\epsilon},p_0) = \frac{1}{2} \sum_i |\bra{\psi_\epsilon} M_i \ket{\psi_\epsilon} - \bra{\psi_0} M_i \ket{\psi_0}| \\
&\le \frac{1}{2} \|\ket{\psi_\epsilon}\!\bra{\psi_\epsilon} - \ket{\psi_0}\!\bra{\psi_0}\|_1  \\
&\le \frac12\cdot 2\cdot \biggl\| \prod_i U_i(\exp\bigl(-it_i(Z+\nu_i P_i+\epsilon X)\bigr)\otimes I) 
    - \prod_i U_i(\exp\bigl(-it_i(Z+\nu_i P_i)\bigr)\otimes I)  \biggr\|. 
\end{align}
The last inequality follows because $\|V\ket{a}\bra{a}V^{\dagger
}-|W\ket{a}\bra{a}W^{\dagger
} \| \le 2\|V-W\|$, where $\ket{a}$ is an arbitrary pure state and $V,W$ are arbitrary unitaries.
By triangle inequality and telescoping sum, we show that 
\begin{align}
    &\biggl\| \prod_i U_i(\exp\bigl(-it_i(Z+\nu_i P_i+\epsilon X)\bigr)\otimes I) 
    - \prod_i U_i(\exp\bigl(-it_i(Z+\nu_i P_i)\bigr)\otimes I)  \biggr\| \\
    &\le \sum_{i=1}^L \|\exp\bigl(-it_i(Z+\nu_i P_i+\epsilon X)\bigr)-\exp\bigl(-it_i(Z+\nu_i P_i)\bigr)\| \\
    &\leq \sum_{i=1}^L 2\epsilon(2\nu_i + \epsilon)t_i +\frac{\epsilon(2\nu_i + \epsilon+1)}{1-\nu_i} \leq \sum_{i=1}^L 2\epsilon(2\nu + \epsilon)t_i +\frac{\epsilon(2\nu + \epsilon+1)}{1-\nu} \\ 
    & \leq 2\epsilon(2\nu + \epsilon)T +\frac{\epsilon(2\nu + \epsilon+1)}{1-\nu}L. 
\end{align}
So, we have 
\begin{align}
&TV(p_{\epsilon},p_0) = \frac{1}{2} \sum_i |\bra{\psi_\epsilon} M_i \ket{\psi_\epsilon} - \bra{\psi_0} M_i \ket{\psi_0}| \\
&\le 2\epsilon(2\nu + \epsilon)\,T
   + \frac{\epsilon(2\nu + \epsilon + 1)}{1-\nu}\,L.
\end{align}

 Hence,
\[
\mathrm{TV}(p_\epsilon, p_0) \le 2\epsilon(2\nu + \epsilon)T +\frac{\epsilon(2\nu + \epsilon+1)}{1-\nu}L,
\]
which is the total variation distance between the measurement outcome distribution over the
two Hamiltonians under a single experiment.
\end{proof}

\subsection{TV upper bound for many experiments}
 We consider the rooted tree representation $\mathcal{T}$ described in \cite{HuangTongFangSu2023learning, huang2022foundations} to model the adaptivity in the choice of experiments. Each node in the tree corresponds to the sequence of measurement outcomes the algorithm has seen so far. At each node $u$, the algorithm runs a single experiment $E_u$ 

\begin{enumerate}
    \item an arbitrary $N_u$-qubit initial state $\ket{\psi_{u,0}} \in \mathbb{C}^{2^{N_u}}$ with an integer $N_u \ge 1$,

    \item an arbitrary POVM $\mathcal{F}_u = \{ M_{u,i} \}_{i=1}^{Q_u}$ with $Q_u$ outcomes on an $N_u$-qubit system,

    \item an $N_u$-qubit unitary of the following form,
    \begin{align}
        &U_{u,L_u+1} \bigl( \exp(-it_{u,L_u}(H+\nu_{u,L_u} P_{u,L_u})) \otimes I \bigr)
        U_{u,L_u} \cdots
        U_{u,3} \\
        &\bigl( \exp(-it_{u,2}(H+\nu_{u,2} P_{u,2})) \otimes I \bigr)
        U_{u,2} \bigl( \exp(-it_{u,1}(H+\nu_{u,1} P_{u,1})) \otimes I \bigr)
        U_{u,1},
    \end{align}
    for some arbitrary integer $L_u$  which represents the number of times the Hamiltonian is evolved, arbitrary evolution times $t_{u,1}, \ldots, t_{u,L_u} \in \mathbb{R}$, $P_{u,1},\ldots, P_{u,L_u}$ are arbitrary single-qubit hermitian operators, $\nu_{u,1},\ldots, \nu_{u,L_u}$ are the corresponding field strengths,  and arbitrary $N_u$-qubit unitaries
    $U_{u,1}, \ldots, U_{u,L_u+1}$. Here $I$ is the identity unitary on $N_u - 1$ qubits.
\end{enumerate}

Each experiment $E_u$ produces a measurement outcome $i \in \{1,\ldots,Q_u\}$, which moves the algorithm from the node $u$ to one of its child nodes. At a leaf node $\ell$, the algorithm stops. By considering the rooted tree representation and allowing the experiment to depend on each node in the tree, we cover all possible learning algorithms that can adaptively choose the experiment that it runs based on previous measurement outcomes.

For each node $u$, we give the following definitions,

\begin{itemize}
    \item $\mathcal{T}_u$ is the subtree with root $u$. 

    \item $p^{(u)}_{\epsilon /0}$ is the distribution over the child nodes of $u$ by considering the probability of moving from $u$
    to that child node under the unknown Hamiltonian $Z+\epsilon X$ or $Z$.

    \item $p^{(\mathcal{T}_u)}_{\epsilon /0}$ is the distribution over the leaf nodes for subtree $\mathcal{T}_u$ by considering the probability of ending
    at that leaf node starting from node $u$ under the unknown Hamiltonian $Z+\epsilon X$ or $Z$. When we start at the root node, we call the distribution $p^{(\mathcal{T})}_{\epsilon /0}$.

    \item $t_{(u)} \triangleq t(E_u) \ge 0$ is the evolution time for the single experiment $E_u$. $L_{u}$ is the number of times the Hamiltonian is applied in this experiment.

    \item $t(\mathcal{T}_u)$ is the maximum of the sum of the evolution time over all paths from root $u$ of the subtree
    $\mathcal{T}_u$ to a leaf node of $\mathcal{T}_u$,
    \begin{equation}
        t(\mathcal{T}_u)
        = \max_{P:\,\text{path on }\mathcal{T}_u} \;
        \sum_{w \in P} t_w.
    \end{equation}
     Similarly, $L(\mathcal{T}_u)$ is the maximum of the sum of number of Hamiltonian applications over all paths from root $u$ of the subtree
    $\mathcal{T}_u$ to a leaf node of $\mathcal{T}_u$,
    \begin{equation}
        L(\mathcal{T}_u)
        = \max_{P:\,\text{path on }\mathcal{T}_u} \;
        \sum_{w \in P} L_w.
    \end{equation}
    When we start at the root node, we call these quantities $t(\mathcal{T})$ and $L(\mathcal{T})$.
\end{itemize}

Because the total evolution time of the learning algorithm is upper bounded by $T_{\max}$, the total
evolution time of the full tree $\mathcal{T}$ satisfies $t(\mathcal{T}) \le T_{\max}$. Similarly we have an upper bound for the total number of Hamiltonian evolutions $L(\mathcal{T}) \le L_{\max}$.\\

The algorithm begins from the root of $\mathcal{T}$. We can establish the total variation upper bound stated in Theorem \ref{theorem:adaptive}. First we build on the tree formalism and state an useful lemma.

\begin{lem} \label{lemma:adaptivebound}
    Consider a rooted tree for an adaptive learning algorithm. We use it for distinguishing the Hamiltonians $Z+\nu P$ and $Z+\nu P+\epsilon X$ , $\epsilon \ge 0$. Every node $u$ in this tree corresponds to a single pre-defined experiment. Then
    \begin{align}
    TV(p_{\epsilon}^{(\mathcal{T})},p_{0}^{(\mathcal{T})}) \le \max_{P:\,\text{path on }\mathcal{T}}\sum_{u\in P} TV(p_{\epsilon}^{(u)},p_{0}^{(u)}).
    \end{align}
\end{lem}
\begin{proof}
    The following inequality was shown in \cite{HuangTongFangSu2023learning}: 
    
    \begin{align}
        1-TV(p_{\epsilon}^{(\mathcal{T}_u)},p_{0}^{(\mathcal{T}_u)}) \ge \Bigl(1-TV(p_{\epsilon}^{(u)},p_{0}^{(u)})\Bigr) \min_{w \in \text{child}(u)}\Bigl(1-TV(p_{\epsilon}^{(\mathcal{T}_{w})},p_{0}^{(\mathcal{T}_{w})}) \Bigr).
    \end{align}
    By simple induction this would imply
    \begin{align}
        1-TV(p_{\epsilon}^{(\mathcal{T})},p_{0}^{(\mathcal{T})}) \ge  \min_{P:\,\text{path on }\mathcal{T}}\prod_{u \in P} \Bigl(1-TV(p_{\epsilon}^{({u})},p_{0}^{({u})}) \Bigr).
    \end{align}
    Now 
    \begin{align}
        \min_{P:\,\text{path on }\mathcal{T}}\prod_{u \in P} \big(1-TV(p_{\epsilon}^{({u})},p_{0}^{({u})}) \big) \ge \min_{P:\,\text{path on }\mathcal{T}} \big(1- \sum_{u \in P}TV(p_{\epsilon}^{({u})},p_{0}^{({u})}) \big).
    \end{align}
    Combining these and simplifying we have
    \begin{align}
    TV(p_{\epsilon}^{(\mathcal{T})},p_{0}^{(\mathcal{T})}) \le \max_{P:\,\text{path on }\mathcal{T}}\sum_{u\in P} TV(p_{\epsilon}^{(u)},p_{0}^{(u)}).
    \end{align}
\end{proof}
\begin{thm} \label{theorem:adaptive}
    Consider the task of deciding whether a given single-qubit Hamiltonian is $Z$ or $Z+\epsilon X$, $\epsilon \ge 0$  where the experimentalist can apply a magnetic field of arbitrary strength at most $0\le \nu<1$ along a direction of choice denoted by a single-qubit traceless hermitian $P$ with eigenvalues $\pm 1$. An adaptive strategy in which the Hamiltonian is evolved no more than $L_{\max}+1$ times and a total evolution time no more than $T_{\max}$, has a total variation distance in the output distribution bounded by 
    \[
    TV(p_{\epsilon}^{(\mathcal{T})},p_0^{(\mathcal{T})}) \le 2\epsilon(2\nu + \epsilon)T_{\max} +\frac{\epsilon(2\nu + \epsilon+1)}{1-\nu}L_{\max}.
    \]
\end{thm} 
\begin{proof}
    Using Lemma \ref{lemma:adaptivebound} for the first inequality, and Theorem \ref{Thm:singl-experimentLB} for the second inequality we have
    \begin{align}
    &TV(p_{\epsilon}^{(\mathcal{T})},p_{0}^{(\mathcal{T})}) \le \max_{P:\,\text{path on }\mathcal{T}}\sum_{u\in P} TV(p_{\epsilon}^{(u)},p_{0}^{(u)})\\
    & \le \max_{P:\,\text{path on }\mathcal{T}}\sum_{u}\epsilon(2\nu + \epsilon)t_{(u)} +\frac{\epsilon(2\nu + \epsilon+1)}{2(1-\nu)}L_{(u)} \\ 
    &\le \epsilon(2\nu + \epsilon)t(\mathcal{T}) +\frac{\epsilon(2\nu + \epsilon+1)}{2(1-\nu)}L(\mathcal{T})\\
    &\le 2\epsilon(2\nu + \epsilon)T_{\max} +\frac{\epsilon(2\nu + \epsilon+1)}{1-\nu}L_{\max}.\label{Eq:many-experimentsUB}
\end{align}
\end{proof}

The two Hamiltonians in the above theorem are slightly different from those in Theorem~2 in the main text in that $Z+\epsilon X$ has not been normalized. However, the above theorem implies Theorem~2 by the simple observation that the normalized Hamiltonian $\hat{n}\cdot \vec{\sigma}$, where $\hat{n}=(\epsilon,0,1)/\sqrt{1+\epsilon^2}$ is very close to $Z+\epsilon X$:
\[
\|\hat{n}\cdot \vec{\sigma} - (Z+\epsilon X)\| \leq \frac{\epsilon^2}{2}.
\]
Because of the above bound the unitary channels generated by evolution under $\hat{n}\cdot \vec{\sigma}$ and $(Z+\epsilon X)$ for time $T$ are only $\epsilon^2 T$ distance away from each other when measured in the diamond norm.
Therefore we have the following corollary by an application of the trace distance triangle inequality:
\begin{cor} \label{cor:adaptive}
    Consider the task of deciding whether a given single-qubit Hamiltonian is $Z$ or $\hat{n}\cdot \vec{\sigma}$ ($\hat{n}=(\epsilon,0,1)/\sqrt{1+\epsilon^2}$, $\vec{\sigma}=(X,Y,Z)$)  where the experimentalist can apply a magnetic field of arbitrary strength at most $\nu<1$ along a direction of choice denoted by a single-qubit hermitian $P$. An adaptive strategy in which there are no more than $L_{\max}+1$ discrete control operations and a total evolution time no more than $T_{\max}$, has a total variation distance in the output distribution bounded by 
    \[
    TV(p_{\epsilon}^{(\mathcal{T})},p_0^{(\mathcal{T})}) \le \epsilon(4\nu + 3\epsilon)T_{\max} +\frac{\epsilon(2\nu + \epsilon+1)}{1-\nu}L_{\max}.
    \]
\end{cor} 

\subsection{Proof of the lower bound}
We finally have all the tools necessary to prove the following theorem, which was informally described in the beginning of this section.
\begin{thm}\label{thm:LB}
    Consider the task of deciding whether a given single-qubit Hamiltonian is $Z$ or $\hat{n}\cdot \vec{\sigma}$ ($\hat{n}=(\epsilon,0,1)/\sqrt{1+\epsilon^2}$, $\vec{\sigma}=(X,Y,Z)$) where the experimentalist can apply a magnetic field of arbitrary strength at most $\nu<1$ along a direction of choice denoted by a single-qubit Hermitian $P$ with eigenvalues $\pm 1$, and a fixed budget of $L$ applications of the Hamiltonian i.e, $L+1, L\in o(\frac{1}{\epsilon})$ discrete control operations.  Then accomplishing this task with some constant probability of success at least $q$ requires evolving the Hamiltonian for a total expected time of at least $T\in \Omega(\frac{1}{4\nu \epsilon+3\epsilon^2})$. 
\end{thm} 
We will use the TV distance upper-bound we derived in the last section. However, note that while we wish to bound the expected time, Eq. \eqref{Eq:many-experimentsUB} addresses the worst-case total time. We remedy this by showing that if the algorithm described in the theorem exists, a bounded-time version should exist as well. We state and prove the following lemma: 
\begin{lem}\label{lemma:boundedtimealgo}
    (bounded time algorithm) Assume there exists an adaptive Hamiltonian learning algorithm that learns the Hamiltonian with $\epsilon$-accuracy with probability $q$ with expected Hamiltonian evolution time at most $T_0$ and expected number of times the Hamiltonian is evolved at most $L_0$ (and therefore $L_0+1$ discrete control operations). Consider a modified algorithm that terminates when evolution time exceeds $kT_0$ ($k>\frac{2}{q}$) or when the number of times the Hamiltonian is applied exceeds $kL_0$, and therefore it has a maximum evolution time of $kT_0$ and maximum number of Hamiltonian evolutions $kL_0$ (i.e, therefore $kL_0+1$ discrete control operations). Then this algorithm succeeds to $\epsilon$-accuracy with probability at least $q'=q-\frac{2}{k}$.
\end{lem}
\begin{proof}
    Note that the probability of success of this bounded time algorithm is equal to the joint probability of following a path with total evolution time less than $kT_0$ as well as succeeding. Since the expected total evolution time is at most $T_0$, the probability associated with paths with evolution time $T_{\text{path}}>kT_0$ must be less than $\frac{1}{k}$ by Markov's inequality. A similar argument implies that the probability associated with paths with number of Hamiltonian evolutions $L_{\text{path}}>kL_0$ must be less than $\frac{1}{k}$.  Then the probability of successful paths in the modified algorithm must be at least $q-\frac{2}{k}$. 
\end{proof}

With the above tools we are ready to prove Theorem~\ref{thm:LB}.
\begin{proof}(for Theorem \ref{thm:LB}) Assume there is an algorithm that can distinguish the two probability
distributions $p_{\epsilon}^{(\mathcal{T})},p_{0}^{(\mathcal{T})}$ with probability at least $q$ with expected total evolution time at most $T_0$ and expected number of Hamiltonian applications at most  $L_0$. Then, according to Lemma \ref{lemma:boundedtimealgo}, there exists a bounded time algorithm with worst-case total evolution time $kT_0$ and $kL_0$ worst-case number of Hamiltonian applications with probability of success at least $q'=q-\frac{2}{k}$ ($k>2/q$).  
Le Cam's two-point method \cite{yu1997assouad} implies
\begin{align}
  &TV(p_{\epsilon}^{(\mathcal{T})},p_{0}^{(\mathcal{T})}) \ge 2q'-1= 2(q-\frac{2}{k})-1.
\end{align}
 
Note the decision tree $\mathcal{T}$ corresponds to the modified bounded time algorithm. Now the upper-bound in Theorem \ref{theorem:adaptive} implies, 
\begin{align}
    &TV(p_{\epsilon}^{(\mathcal{T})},p_{0}^{(\mathcal{T})}) \le \epsilon(4\nu + 3\epsilon)t(\mathcal{T}) +\frac{\epsilon(2\nu + \epsilon+1)}{1-\nu}L(\mathcal{T}) \\
    &\le \epsilon(4\nu + 3\epsilon)kT_{0} +\frac{\epsilon(2\nu + \epsilon+1)}{1-\nu}kL_{0}
\end{align}
The final inequality follows from how we defined the modified bounded time algorithm. Putting these together,
\begin{align}
  \epsilon(4\nu + 3\epsilon)kT_0 +\frac{\epsilon(2\nu + \epsilon+1)}{1-\nu}kL_0\ge TV(p_{\epsilon}^{(\mathcal{T})},p_{0}^{(\mathcal{T})}) \ge 2(q-\frac{2}{k})-1.
\end{align}
Combining the equations and recalling that $L_0 \in o(\frac{1}{\epsilon})$, 
\begin{align}
    T_0 \in \Omega(\frac{1}{4\nu\epsilon+3\epsilon^2}). \label{Eq:time-LB}
\end{align}
\end{proof}
This theorem establishes the lower-bound in Theorem~\ref{thm:LB}.



\subsection{Bound for \texorpdfstring{$|w(\epsilon)-w(0)|$}{|w(epsilon)-w(0)|}}\label{appendix:difference}
In this Appendix we show that 
\begin{align}
    |w(\epsilon)-w(0)| \le \epsilon (2\nu +\epsilon),
\end{align}
for $\nu < 1$. Recall that $w(\epsilon)=\|(1+\nu_z)Z+\nu_y Y+(\nu_x+\epsilon)X\|=\sqrt{(1+\nu_z)^2+\nu_y^2+(\nu_x+\epsilon)^2}$, and $\sqrt{\nu^2_x+\nu^2_y+\nu^2_z} \le \nu.$
\begin{align}
    &|w(\epsilon)-w(0)|\\
    & =\frac{|w^2(\epsilon)-w^2(0)|}{w(\epsilon)+w(0)} \le \frac{|w^2(\epsilon)-w^2(0)|}{\bigl(w^2(\epsilon)+w^2(0) \bigr)^\frac12}\\
    &\le \frac{|(\nu_x+\epsilon)^2-\nu_x^2|}{\bigl(2(1+\nu_z)^2+2\nu_y^2+(\nu_x+\epsilon)^2+ \nu_x^2\bigr)^\frac{1}{2}}\\
    &\le \frac{|(2\nu_x+\epsilon)\epsilon|}{\bigl(2(1+\nu_z)^2+\frac{\epsilon^2}{2}+\frac12(2\nu_x+\epsilon)^2\bigr)^\frac{1}{2}}\\
\end{align}
Lets first consider the case where $\nu \le 1$.
\begin{align}
    &\le \epsilon.\frac{|(2\nu_x+\epsilon)|}{\bigl(2(1-\nu)^2+\frac12(2\nu_x+\epsilon)^2\bigr)^\frac{1}{2}}\\
    &\le \epsilon.\frac{|2\nu_x+\epsilon|}{\bigl(1-\nu+\frac12|2\nu_x+\epsilon|\bigr)}.
\end{align}
The second inequality follows from $\bigl(2a^2+2b^2)^\frac12 \ge |a|+|b|$, and $|1-\nu|=1-\nu$ since $\nu < 1$. Next we use the fact that if $a,b >0,a'>a$ then $\frac{a}{a+b} \le \frac{a'}{a'+b}$ and than $|\nu_x|\le \nu.$ 
\begin{align}
    &\le \epsilon.\frac{|2\nu_x+\epsilon|}{\bigl((1-\nu)+\frac12|2\nu_x+\epsilon|\bigr)}\le \epsilon.\frac{|2\nu+\epsilon|}{\bigl((1-\nu)+\frac12(2\nu+\epsilon)\bigr)}\\
    &\le \epsilon.\frac{|2\nu+\epsilon|}{1+\frac{\epsilon}{2}} \le \epsilon(2\nu + \epsilon).
\end{align}


\end{document}